\documentclass[letterpaper,11pt]{article}
\usepackage[utf8]{inputenc}
\usepackage[margin=1in]{geometry}
\usepackage{graphicx} 
\usepackage{fullpage}
\usepackage[style=alphabetic,natbib=true,maxbibnames=9,maxalphanames=9,maxcitenames=9]{biblatex}
\usepackage{xcolor}
\usepackage{multirow}
\usepackage{booktabs}
\usepackage{hyperref}
\usepackage{amsmath,amsthm,thm-restate}
\usepackage{amssymb}
\usepackage{algorithmicx}
\usepackage[ruled, noend, linesnumbered]{algorithm2e}
\usepackage{float}
\usepackage{mathtools}
\usepackage{tcolorbox}
\usepackage{enumitem}
\usepackage{setspace}
\usepackage{ulem}
\usepackage{subcaption}
\definecolor{mygreen}{RGB}{10,150,110}
\definecolor{myred}{RGB}{150,10,20}
\hypersetup{
     colorlinks=true,
     citecolor= mygreen,
     linkcolor= myred
}
\providecommand{\email}[1]{\href{mailto:#1}{\nolinkurl{#1}\xspace}}

\usepackage[noabbrev,nameinlink,capitalize]{cleveref}
\crefname{lemma}{Lemma}{Lemmas}
\crefname{theorem}{Theorem}{Theorems}
\crefname{property}{Property}{Properties}
\crefname{claim}{Claim}{Claims}
\crefname{result}{Result}{Results}
\crefname{definition}{Definition}{Definitions}
\crefname{observation}{Observation}{Observations}
\crefname{proposition}{Proposition}{Propositions}
\crefname{assumption}{Assumption}{Assumptions}
\crefname{line}{Line}{Lines}
\crefname{figure}{Figure}{Figures}
\creflabelformat{property}{(#1)#2#3}
\crefname{equation}{}{}
\crefname{section}{Section}{Sections}
\crefname{appendix}{Appendix}{Appendices}
\crefname{algCounter}{Algorithm}{Algorithms}
\Crefname{algCounter}{Algorithm}{Algorithms}

\newtheorem{lemma}{Lemma}[section]
\newtheorem{theorem}[lemma]{Theorem}
\newtheorem{proposition}[lemma]{Proposition}

\newtheorem{definition}[lemma]{Definition}

\newtheorem{fact}[lemma]{Fact}

\newtheorem{remark}[lemma]{Remark}
\newtheorem*{remark*}{Remark}

\newcommand{\opt}{\textnormal{OPT}}
\newcommand{\sol}{\textnormal{SOL}}

\newcommand{\poly}{\operatorname{poly}}

\newcommand{\argmin}{\textnormal{argmin}}

\newcommand{\RGMIS}{\ensuremath{\textup{RGMIS}}}
\newcommand{\AlgAddMul}{\ensuremath{\textup{AlgAddMul}}}
\newcommand{\AlgMul}{\ensuremath{\textup{AlgMul}}}
\newcommand{\Active}{\ensuremath{\textup{Active}}}
\newcommand{\CurrentMISSize}{\ensuremath{\textup{CurrentMISSize}}}
\newcommand{\cA}{\mathcal{A}}
\newcommand{\cC}{\mathcal{C}}
\newcommand{\cO}{\mathcal{O}}
\newcommand{\tO}{\widetilde{O}}
\newcommand{\E}{\mathbb{E}}
\newcommand{\Prob}{\mathbb{P}}
\newcommand{\eps}{\varepsilon}
\newcommand{\SF}[0]{\ensuremath{\textsf{SF}}}
\newcommand{\TMIS}{\ensuremath{T_\mathrm{MIS}}}
\newcommand{\cmis}{\ensuremath{c_\mathrm{MIS}}}
\def\target{\mbox{Target}}

\title{Sublinear Metric Steiner Forest via Maximal Independent Set}
\author{Sepideh Mahabadi\thanks{Microsoft Research; \email{smahabadi@microsoft.com}.} \and Mohammad Roghani\thanks{Stanford University. E-mail: \email{roghani@stanford.edu}. The work was done while the author was an intern at Microsoft Research.} \and Jakub Tarnawski\thanks{Microsoft Research; \email{jakub.tarnawski@microsoft.com}.} \and Ali Vakilian\thanks{Virginia Tech; \email{vakilian@vt.edu}.}}
\date{}

\begin{document}
\begin{titlepage}
\maketitle
\thispagestyle{empty}
\begin{abstract}
    In this work we consider the Metric Steiner Forest problem in the sublinear time model.
    Given a set $V$ of $n$ points in a metric space where distances are provided by means of query access to an $n\times n$ distance matrix, along with a set of $k$ terminal pairs $(s_1,t_1), \dots, (s_k,t_k)\in V\times V$, the goal is to find a minimum-weight subset of  edges that connects each terminal pair.
    Although sublinear time algorithms have been studied for estimating the weight of a minimum spanning tree in both general and metric settings, as well as for the metric Steiner Tree problem, no sublinear time algorithm was known for the metric Steiner Forest problem.
    
    Here, we give an $O(\log k)$-approximation algorithm for the problem that runs in time $\widetilde{O}(n^{3/2})$.
    Along the way, we provide the first sublinear-time algorithm for estimating the size of a Maximal Independent Set (MIS). Our algorithm runs in time $\widetilde{O}(n^{3/2}/\varepsilon^2)$
    under the adjacency matrix oracle model
    and obtains a purely multiplicative $(1+\varepsilon)$-approximation.
    Previously, sublinear-time algorithms for MIS were only known for bounded-degree graphs.
\end{abstract}
\end{titlepage}

\setcounter{page}{1}
\section{Introduction}
{\em Steiner Forest} is a classic network design problem where, given a weighted undirected graph $G = (V, E, w)$ where the weights are specified
by $w:E\rightarrow \mathbb{R}_{\geq 0}$, along with 
a set of $k$ source-sink terminal pairs $T=\{(s_1,t_1),\dots,(s_k, t_k)\}$, 
the goal is to find a subgraph $G'$ of minimum total weight such that each pair $(s_i, t_i)$ is connected in $G'$. 
This problem generalizes the {\em Steiner Tree} problem (in fact it is also known as {\em generalized Steiner tree problem}), and hence is APX-hard. In particular, approximating it to a factor better than 96/95 is NP-hard \cite{chlebik2008steiner}.

The first approximation algorithm for the problem was a primal-dual based approach given by Agrawal, Klein, and Ravi \cite{agrawal1995trees} that achieved a $2$-approximation. Later, Goemans and Williamson \cite{goemans1995general} provided a simplified simulation of their primal-dual algorithm, which gives a $(2 - 2
/n)$-approximate solution, where $n$ is the number of vertices.
Improving the approximation guarantee of 2 has remained an open problem until very recently, when~\cite{ahmadi2025breaking} broke this barrier by designing a $2-2^{-11}$-approximation.
Furthermore, Gupta and Kumar \cite{gupta2015greedy} provided a simple greedy-based algorithm achieving a constant larger than 2. Designing algorithms with specific features and addressing the problem within various models such as online, parallel, and streaming remains an active area to day  (see e.g.~\cite{gross2018local, czumaj2024streaming, jin2024streaming, chekuri2025streaming, bamas2022improved, ghalami2022parallel}).

There is a large body of work on designing sublinear time algorithms for graph problems such as minimum spanning tree (MST)~\cite{chazelle2005approximating}, maximal independent set (MIS)~\cite{YoshidaYISTOC09}, matching~\cite{YoshidaYISTOC09,  kapralov2020space, Behnezhad21,BehnezhadRR23a, BehnezhadRRS-SODA23, BhattacharyaKS23, abrrfocs25, MahabadiRT25, BehnezhadRR24, BehnezhadRR23b, BhattacharyaKS-STOC23}, spanners~\cite{lenzen2018centralized,parter2019local,arviv2023improved}, metric MST~\cite{czumaj2009estimating}, metric Steiner Tree~\cite{chen2023query, mahabadi2025sublinear}, and metric TSP~\cite{chen2020, ChenMetric-Arxiv22, TSP-icalp24}, among others. Given that a sublinear time algorithm cannot afford to read the entire graph, it is instead provided an \emph{oracle access} to the input graph. 
There are two main oracle models for graph problems considered in the literature.
In the {\em adjacency list oracle}, the algorithm can query $(v,i)$, where $v\in V$ and $i\leq n$, and the oracle reports the $i$-th neighbor of the vertex $v$ in its adjacency list (along with the weight of the edge in case of weighted graphs), or NULL if $i$ is larger than the number of $v$'s neighbors. In the {\em adjacency matrix oracle}, the algorithm can query $(u,v)$, where $u,v\in V$, and the oracle reports whether there exists an edge between $u$ and $v$ (along with its weight in the case of weighted graphs).

The focus of this work is on designing a {\em sublinear time} algorithm for the metric Steiner Forest problem under the adjacency matrix model. 
We remark that for ``non-metric'' instances, even approximating the MST cost within any  factor requires $\Omega(n^2)$ time. Consider two vertex sets $S$ and $T$, each of size $n/2$, where all pairwise distances between vertices within $S$ or within $T$ are zero, and all distances between $S$ and $T$ are one, except for a single random pair $(s^*, t^*)$ across $S$ and $T$, whose distance is zero with probability $1/2$. Any algorithm that approximates the MST cost with high probability must distinguish between a total cost of zero and one. This requires identifying the pair $(s^*, t^*)$, which in turn requires $\Omega(n^2)$ queries.

\begin{definition}[Sublinear Metric Steiner Forest]\label{def:steiner-forest}
In the metric Steiner Forest problem, we are given a set of points $V$, a set of $k$ terminal pairs $T=\{(s_1,t_1),\dots,(s_k,t_k)\}\subseteq V\times V$, and query access to the $|V|\times |V|$ distance matrix of a metric space $(V, w)$, where an oracle query for $(u,v)$ returns the weight $w(u,v)$ of the edge $(u,v)$.

Let $\SF(V, T, w)$ denote the minimum weight of a Steiner Forest on instance $(V, T, w)$.    
Then, the goal is to design an algorithm that estimates $\SF(V, T, w)$ using $o(n^2)$ queries to the distance matrix via the oracle.
\end{definition}

\paragraph{First contribution.} In this work, we give the first sublinear time algorithm for approximating $\SF(V,T,w)$. In Section \ref{sec:steiner-forest} we show the following result.

\begin{theorem}\label{thm:steiner-forest-main}
There exists an algorithm for estimating the cost of Steiner Forest within a multiplicative factor of $O(\log k)$ using $O(\TMIS \cdot \log k) = \tO(n^{3/2})$ queries to the distance matrix oracle. Here, $\TMIS$ denotes the best runtime of a sublinear algorithm for finding a multiplicative $O(1)$-approximation for the size of any MIS under the adjacency matrix model.
\end{theorem}

We assume that the number $k$ of terminal pairs is $O(n)$. 
Note that for any given Steiner Forest instance, there exists a list of at most $n-1$ terminal pairs $(s_i,t_i)$ that fully characterize it.
Without this assumption, an $O(k)$ term will be added to the runtime of \cref{thm:steiner-forest-main} to account for reading the pairs and computing this succinct representation. Moreover, this dependence on $k$ in the runtime is required. To see this, consider the case where there is a single vertex $v$ that is far away from the rest of the graph, and there only exists a single pair involving $v$. Unless this pair is detected, the reported solution will be off by an arbitrarily large factor. 

\paragraph{Maximal Independent Set (MIS).}
As a key component of our approach,
we develop a sublinear algorithm for the 
maximal independent set (MIS) problem
in the adjacency matrix model.

An independent set is a set of vertices such that no two of its elements are connected by an edge.
While the \emph{maximum} independent set problem is NP-hard to even approximate within a factor $n^{1-\eps}$ \cite{hastad1996clique,zuckerman2006linear},
for our purposes we only require
an independent set
that is \emph{maximal}.

MIS is a basic problem in its own right.
It is also broadly used as a building block
and studied
in computational models with limited time, memory or communication,
such as distributed and parallel \cite{luby1985simple,alon1986fast,BlellochFS12,ghaffari2016improved,FischerA19,GhaffariH21},
Massively Parallel Computing \cite{LattanziMSV11,BrandtFU18,Onak18,ghaffari2018improved,GhaffariU19,GhaffariGJ20},
Local Computation Algorithms \cite{NguyenOnakFOCS08,YoshidaYISTOC09,RubinfeldTVX11,alon2012space,ReingoldV16,LeviRY17,ghaffari2022local},
dynamic~\cite{BehnezhadDHSS19,AssadiOSS19,ChechikZ19},
as well as in algorithms for other fundamental problems,
such as matching (a maximal matching is an MIS of the corresponding line graph) and
vertex cover~\cite{YoshidaYISTOC09,OnakSODA12,Behnezhad21}, or correlation clustering~\cite{AilonCN05,BehnezhadCMT22,DalirrooyfardMS24}.

Since an MIS cannot be explicitly found in sublinear time\footnote{Consider, for example, an instance with only one edge connecting two random vertices. To find any MIS, one must locate that edge, which requires $\Omega(n^2)$ time in the adjacency matrix model.},
the task becomes that of estimating the MIS size,
or developing oracles that efficiently check whether a vertex is in some fixed MIS.
The study of such algorithms was initiated by Nguyen and Onak~\cite{NguyenOnakFOCS08}
in the context of bounded-degree graphs.
They considered the \textbf{random greedy maximal independent set (RGMIS)} process,
which iterates over vertices according to a random permutation $\pi$
and adds a vertex to the set $\RGMIS(\pi)$ if none of its predecessors were already added.
A natural oracle to check whether a given vertex $v$ is in $\RGMIS(\pi)$
is to ask, for all its neighbors $u$ with lower rank in $\pi$, whether $u$ is in $\RGMIS(\pi)$,
and return yes only if none of them are.

Nguyen and Onak~\cite{NguyenOnakFOCS08}
showed that
for a random permutation $\pi$,
the expected query complexity of this oracle is $2^{O(\Delta)}$,
where $\Delta$ is the maximum degree.
(They then used this to obtain a $(2,\eps n)$-approximation\footnote{We use this notation to say that the algorithm has a multiplicative approximation ratio of $2$, with an additional additive error of up to $\eps n$.} for maximum matching in time $2^{O(\Delta)}/\eps^2$.)
They conjectured that if the neighbors $u$ are queried in the order of increasing rank in $\pi$ (see \cref{alg:abstract-oracle}),
the query complexity bound can be improved.

In a seminal result, Yoshida, Yamamoto and Ito~\cite{YoshidaYISTOC09}
showed that if the query vertex $v$ is also random,
this improved oracle indeed has an expected query complexity of only $O(\Delta)$
(see \cref{thm:yyi} for the precise statement).
By identifying the low-rank neighbors of a vertex
in time $O(\Delta)$ in the adjacency list model
or in time $O(n)$ in the adjacency matrix model,
this implies an algorithm for the vertex oracle
that runs in time
$O(\Delta^2)$
in the adjacency list model
or in time
$O(\Delta n)$
in the adjacency matrix model.
We note that these are not sublinear-time if $\Delta = \Omega(n)$, and
they only return the answer for a single vertex query. To the best of our knowledge,
no faster algorithms for MIS itself have been developed.
However, the result of Yoshida, Yamamoto and Ito~\cite{YoshidaYISTOC09} has been highly influential for other basic problems, particularly maximum matching.
By studying line graphs, they obtained a $(2,\eps n)$-approximation algorithm that runs in time $\poly(\Delta)/\eps^2$.
This was further improved by Onak, Ron, Rosen and Rubinfeld~\cite{OnakSODA12}
and Behnezhad~\cite{Behnezhad21}.

\paragraph{Second contribution.} As the second main result of our paper, we give a sublinear time algorithm to estimate the RGMIS size in the adjacency matrix model.
We note that we obtain a purely multiplicative approximation, without an additive error; this is indeed required for our approximation guarantee for the Steiner Forest problem.

\begin{restatable}{theorem}{mismainthm} \label{thm:mis-main}
    For any $\eps \in (0,1)$
    there is an algorithm (\cref{alg:mainmis}) that,
    given a graph $(V,M)$
    with oracle access to its adjacency matrix,
    %
    %
    with high probability
    reports a multiplicative $(1+\eps)$-approximation to the value $|\RGMIS(\pi)|$
    for some permutation $\pi$ of $V$
    and
    runs in $\tO(n^{3/2}/\eps^2)$ time.
\end{restatable}


\subsection{Related Work}
\paragraph{Sublinear MST.}
The problem of estimating the weight of the minimum spanning tree (MST) of a given graph in sublinear time was first studied by Chazelle, Rubinfeld, and Trevisan \cite{chazelle2005approximating}, who considered the adjacency list model. They gave an algorithm with a running time of $\tO(\bar{d}\cdot W \cdot \eps^{-2})$ providing a $(1+\eps)$-approximation for graphs with weights in $[1,W]$ and average degree $\bar{d}$. They further showed a lower bound of $\Omega(\bar{d} \cdot W\cdot \eps^{-2})$ on the time complexity of $(1+\eps)$-approximation algorithms.
  
\paragraph{Sublinear Metric MST.}
Czumaj and Sohler \cite{czumaj2009estimating} considered the metric version of the problem, where the vertices lie in a metric space and the weight of each edge corresponds to the distance between two points in this space. These distances can be queried in $O(1)$ time. They provided a $(1+\eps)$-approximation algorithm with running time $\tO(n/\eps^{7})$. Since the full description of an $n$-point metric space is of size $\Theta(n^2)$, this algorithm is sublinear with respect to the input size. However, outputting an actual spanning tree whose weight is within any constant factor of the minimum spanning tree was shown to require $\Omega(n^2)$ time \cite{indyk1999sublinear}.

By allowing more sophisticated orthogonal range queries and cone approximate nearest neighbor queries, \cite{czumaj2003sublinear} designed a $(1+\varepsilon)$-approximation algorithm for estimating the minimum weight of a spanning tree in constant-dimensional Euclidean space, where the input is supported by a minimal bounding cube enclosing it. The algorithm runs in time $\tO(\sqrt{n} \cdot \poly(1/\eps))$.


\paragraph{Sublinear Metric Steiner Tree.}
For metric Steiner Tree, in which only a subset of vertices of size $k$, i.e., terminals, need to be connected, the algorithm of \cite{czumaj2009estimating} gives a $(2+\eps)$ approximation in time $O(k/\eps^7)$. This uses the well-known result by Gilbert and Pollak \cite{gilbert1968steiner} showing that an $\alpha$-approximation for MST over the metric induced on the terminals is a $(2\alpha)$-approximation for the metric Steiner Tree problem. Later, Chen, Khanna, and Tan \cite{chen2023query} studied the problem with strictly better-than-2 approximation.
On the lower bound side, they showed that for any $\eps > 0$, estimating the Steiner Tree cost to within a $(5/3 -\eps)$-factor requires $\Omega(n^2)$ queries, even when the number of terminals is constant. Moreover, they showed that for any $\eps > 0$, estimating the Steiner Tree cost to within a $(2 - \eps)$-factor requires $\Omega(n + k^{6/5})$ queries.
Additionally, they proved that for any $0 < \eps < 1/3$, any algorithm that outputs a $(2-\eps)$-approximate Steiner Tree (not just its cost) requires $\Omega(nk)$ queries. 
On the upper bound side, they showed that an algorithm that, with high probability, computes a $(2-\eta)$-approximation of
the Steiner Tree cost using $\tO(n^{13/7})$ queries, where $\eta > 0$ is a universal constant. Recently, \cite{mahabadi2025sublinear} gave an improved algorithm with query complexity of $\tO(n^{5/3})$ achieving such a $(2-\eta)$-approximation.
We emphasize that the algorithms of \cite{chen2023query} and \cite{mahabadi2025sublinear} and the challenges they face are inherently different from ours.
They start with the MST solution on the set of terminals, which already gives a simple factor-$2$ approximation, and then decide whether it is possible to improve over it by adding Steiner (non-terminal) vertices (which is done by reducing to an instance of the set cover problem). However, for the Steiner Forest problem such a simple $2$-approximate solution is not known and thus a different set of techniques needs to be developed.

\subsection{Overview of our Algorithms}
A natural framework one would first try for solving this problem is the approach used in \cite{chazelle2005approximating, czumaj2009estimating} for estimating the weight of an MST; as we describe below, it fundamentally does not suffice for the Steiner Forest problem.
Consider the subgraph consisting of all edges of weight up to some threshold $\tau$, and let $c$ be the number of connected components in this graph. Then any MST needs to use {\em exactly} $c-1$ edges of weight larger than $\tau$. This intuition was formalized in these works, giving the equality $\mbox{MST}\approx n-W+\eps\sum_{i=0}^n (1+\eps)^i c_i$, where $c_i$ is the number of connected components of the graph where we set the threshold to $\tau_i=(1+\eps)^i$. Then these prior works focus on estimating $c_i$ and providing a rigorous proof that the approximation they can get for $c_i$ is sufficient for the ultimate task of estimating MST.

\medskip

There is an equivalent intuition for the Steiner Forest problem when we consider the threshold graph parameterized by $\tau$, with two differences.
\begin{itemize}
    \item First, $c$ should now be the number of {\em active} connected components; a component is active if for some pair $(s_j,t_j)$ one of the two vertices is inside the component and the other one is outside.
    \item Second, the number of edges of weight more than $\tau$ used in the optimal solution is now at least $c/2$ and at most $c-1$, because not all active components will need to be connected in the end -- we might only need a matching between these active components.
\end{itemize}

Thus, a potential algorithm would be to instead estimate the number of active connected components $c_i$ for each threshold $\tau_i$ and follow the approach of \cite{czumaj2009estimating}. Note that, roughly speaking, the second item above does not cause much trouble, as it only results in a factor-2 blow-up in the approximation.

\paragraph{Challenge in estimating the number of active components.} A common approach for counting the number of connected components in a graph is to sample a random vertex and start growing a ball from that vertex (e.g.~using BFS) until either the entire component is visited, in which case we increment our estimate of the number of components, or the total number of visited vertices reaches a threshold, in which case we can ignore that component and show that the error introduced by ignoring such large components is small. While this is the core idea, the details required to obtain a multiplicative approximation in \cite{czumaj2009estimating} are more involved. In particular, they run a modified BFS procedure where very close vertices are treated as a single vertex, allowing for further improvements.

To mimic this approach, first of all, one has to sample a random terminal as opposed to a random vertex, to ensure that unrelated parts of the graph are not selected most of the time. However, one challenge that does not allow us to follow this approach is that when growing a ball from a terminal (or a vertex in general), the total time spent depends on the {\em volume} of the component that we explore, i.e., the number of vertices we visit. 
For the MST problem, this volume can be related to the cost that an optimal solution must pay within that component, since ultimately all vertices must be connected.

However, this does not hold for the Steiner Forest problem. When we grow a ball around a terminal, the optimal solution might only select a single path from the center to the boundary of the ball to connect the terminal to its pair.
At a high level, the cost of such search algorithms depends on the {\em volume} of the ball, whereas the optimal solution for Steiner Forest might only pay the cost of the ball's {\em radius}.

\paragraph{Our new approach.}
In this work, we adopt a different approach than counting the number of active components, which follows a similar intuition but can be implemented in sublinear time. 
We consider balls of radius $\tau$ around all terminals, and call a ball {\em active} if the matching vertex of the center of the ball falls outside of it.
The advantage is that checking whether a ball is active or not under this new definition requires only a single query. 
Now, if we take a set of $M$ disjoint active balls, then any optimal solution needs to pay a cost of at least $M\cdot \tau$ (to buy a path from the center to the boundary of each ball).

Therefore, instead of counting the number of active components, our algorithm seeks to find a {\em maximal independent set} on the set of active balls for a certain threshold $\tau$. This is where we use the sublinear algorithm for MIS on the graph defined on the set of active balls, and output $\sum_i M_i\cdot\tau_i$ for exponentially increasing thresholds $\tau_i$. We show that this achieves an $O(\log k)$-approximation for Steiner Forest.

To that end,
we note that each terminal in the MIS $M_i$ induces a certain cluster of low diameter $O(\tau_i)$.
We show how to inductively construct a solution that internally connects every such cluster,
for increasing~$i$,
while paying at most $M_i \cdot \tau_i$ for each threshold $\tau_i$.
This hierarchically built partial solution
then serves as a scaffolding to connect the demands.
We assign each demand pair carefully to a specific level~$i$.
Then, for each $i$, we show that one can connect all assigned demands 
within the same budget of $M_i \cdot \tau_i$
by selecting a spanning forest in a suitably defined auxiliary graph.
%

Finally, we also show that this MIS-based approach cannot lead to a better than $O(\log k)$-approximation for Steiner Forest.

\paragraph{Challenges for MIS.} Let us now describe the main challenges that arise when developing a sublinear-time algorithm for MIS, and outline our key ideas for addressing them.

Recall from the introduction above
that a direct implementation of the RGMIS oracle~\cite{YoshidaYISTOC09}
in the adjacency matrix model
would result in an expected runtime of $O(\Delta n)$
for a single vertex query.
The two major issues are that
this is not sublinear-time for dense graphs,
and that it only returns a single bit of information about a single vertex,
whereas the MIS size can range from $1$ to $n$.
We can thus obtain a $(1,\eps n)$  multiplicative-additive approximation
by making $\widetilde{O}(1/\eps^2)$ queries,
but a purely multiplicative approximation seems to require $\Omega(n)$ queries
(to e.g. differentiate between different constant MIS sizes).

Even in the adjacency list model, the query complexity analysis of \cite{YoshidaYISTOC09} does not explicitly translate into an algorithm. The main challenge lies in generating the neighbors of a vertex in increasing order of their ranks in the random permutation $\pi$. Achieving this efficiently seems difficult if the goal is to spend only $o(n)$ time per step to obtain the next neighbor in the ordered list. Moreover, approaches that rely on this query complexity analysis to design algorithms for maximum matching \cite{Behnezhad21} cannot be applied here. In the case of matching, each edge has only two endpoints, and revealing its rank requires ``notifying" just these two endpoints. However, for MIS, we would need to notify up to $\Omega(n)$ neighbors in the worst case.

\paragraph{Our MIS approach.}
As a warm-up let us first describe a simpler approach that leads to an algorithm with $\tO(n^{5/3})$ runtime.
Our first insight is that,
rather than using an oracle as above,
an RGMIS can be also built directly bottom-up by iterating over vertices in the permutation $\pi$;
when a vertex is visited, we add it to the RGMIS and disqualify all its neighbors from further consideration.
We can perform $s$ steps of this process in time $O(ns)$,
and it will result in explicitly building the entire RGMIS if its size is at most $s$.
Let us run it for $s=\tO(n^{2/3})$ steps.
One can then prove that with high probability, the maximum degree in the remaining graph is at most $n^{1/3}$;
thus the vertex oracle~\cite{YoshidaYISTOC09} runs in time $O(n^{4/3})$.
Then, $\tO(n^{1/3}/\eps^2)$ queries to this oracle will result in an $(1+\eps, O(\eps n^{2/3}))$ multiplicative-additive error, which gives a purely multiplicative approximation for the entire graph because the RGMIS size is now known to be at least $n^{2/3}$.

In order to improve upon the above and obtain our final running time of $\tO(n^{3/2})$,
we show a different implementation of the RGMIS vertex oracle in the adjacency matrix model
(\cref{alg:O}),
and
we prove that it has an expected running time of $O(n)$ instead of $O(\Delta n)$,
even in dense graphs.
We do this via a probabilistic argument in which
we couple the execution of our \cref{alg:O}
on the original input graph
to the execution of the Yoshida, Yamamoto and Ito oracle (\cref{alg:abstract-oracle})
on a certain larger graph,
and relate the number of adjacency matrix oracle calls of the former
to the number of recursive calls of the latter.
This is possible as long as the random permutation on this larger graph
has a favorable structure;
we show that this happens with at least a constant probability
via a connection to a combinatorial ballot voting problem
whose study dates back to the 19th century~\cite{whitworth,bertrand1887solution}.

\section{Sublinear Metric Steiner Forest}\label{sec:steiner-forest}
In this section, we will prove \cref{thm:steiner-forest-main}. We assume that we have access to an algorithm for MIS that gives a multiplicative $O(1)$-approximation in the adjacency matrix model. \cref{thm:mis-main}, shown in \cref{sec:MIS}, provides such an algorithm.

\paragraph{Notation.}
We will use the term {\em terminal} to refer to any of the vertices in $\bigcup_{i\leq k} \{s_i, t_i\}$, and will abuse notation and use $T$ to also refer to all terminals when clear from the context. We use {\em terminal pairs} to refer specifically to the pairs $(s_i,t_i)$. For any terminal in $T$, we define its {\em match} as the other vertex in the pair, i.e., $m(s_i)=t_i$ and $m(t_i)=s_i$.
Moreover, for ease of notation, we assume that terminal pairs are disjoint. Note that this is without loss of generality by creating copies of original vertices, where each pair $(s_i, t_i)$ uses distinct copies of the corresponding two vertices.

Finally, we use $\opt$ to denote the value of the optimal Steiner Forest connecting all given terminal pairs.

\paragraph{Preprocessing.}
First, we apply a standard preprocessing step to bound the approximation factor in terms of $\log k$ rather than $\log W$. Let $X = \max_{i\leq k} w(s_i,t_i)$, which can be computed using $O(k)$ queries. It is straightforward to show that $X \leq \opt \leq k\cdot X$. Now, we will ignore all terminal pairs $(s_i,t_i)$ such that $w(s_i,t_i)\leq X/k$. At the end of the algorithm we can connect these pairs directly; since there are at most $k$ such pairs, their additional cost will be at most $X$. Thus, they contribute only an additive term of $1$ to the approximation factor. Also note that we do not need to consider using edges whose weights are larger than $kX$. By scaling, from now on we will assume that the minimum terminal pair has distance $2$, and thus the maximum weight of an edge that we would want to ever pick in our solution is at most $2k^2$.

\paragraph{Threshold graph.} Our algorithm will consider several thresholds $\tau_i = 2^i$ for $i$ from $0$ to $L=\log (2k^2) = O(\log k)$. We will use the {\em threshold graph} $G_i=(V,E_i)$ to refer to an unweighted subgraph of $G$ where $(u,v)\in E_i$ if $w(u,v)\leq \tau_i$.

\paragraph{Active terminals and ball graphs.} We will say that a terminal $s\in T$ is {\em active} in level $i\leq L$ if $w(s,m(s))\geq \tau_i$. Now we define an unweighted {\em ball graph} $H_i=(V_i^H,E_i^H)$, where $V_i^H\subseteq T$ is the set of active terminals in level $i$. More precisely, we are implicitly thinking of $V_i^H$ as the set of balls of radius $\tau_i$ around each active terminal. Now two vertices $u,v\in V_i^H$ have an edge between them if their corresponding balls collide, i.e., $(u,v)\in E_i^H$ iff $w(u,v)<2\tau_i$. 

\paragraph{Algorithm.}
Now our algorithm is as follows. For each $0\leq i\leq L$, let $M_i$ be the size of {\em any} MIS in $H_i$, and let $\hat{M}_i$ be a multiplicative estimate of it
given by the algorithm of \cref{thm:mis-main},
i.e., $M_i \leq \hat{M}_i \leq \cmis\cdot M_i$
(we can set $\cmis$ to be e.g.~$1.01$).
Our algorithm will output $\sol=\sum_{i=0}^L \hat{M}_i\cdot \tau_i$.


\begin{lemma}\label{lem:lower-bound-opt}
    For any $i\leq L$, we have $\opt \geq M_i\cdot \tau_i$.
\end{lemma}
\begin{proof}
    Consider an MIS in $H_i$ whose size is $M_i$; let us denote it by $U_i\subseteq V_i^H$. Now take any $u\in U_i$ and consider the ball $B_u$ of radius $\tau_i$ around $u$. Given that $u$ is active in level $i$, it means that in order to reach from $u$ to its match $m(u)$, intuitively, any Steiner Forest solution has to pick a path of length at least $\tau_i$ inside $B_u$.

    More concretely,
    for each $u$ we consider any path $P_u$ connecting $u$ and $m(u)$ in the solution.
    We would like to say that $P_u$ contains a subpath of length at least $\tau_i$ that lies inside $B_u$. Since these balls are disjoint (as $U_i$ is an MIS in $H_i$, the distance between any two ball centers is at least $2\tau_i$), so are the subpaths, which implies that we get $\opt\geq M_i\cdot \tau_i$.

    However, strictly speaking $P_u$ might not contain such a subpath, as we are dealing with a discrete metric space.
    For such $u$,
    for the sake of this analysis only,
    we consider the first edge $(v_1,v_2)$ of $P_u$ that exits $B_u$ (i.e., $w(u,v_1)<\tau_i$ but $w(u,v_2)\geq\tau_i$). 
    We subdivide the edge $(v_1,v_2)$ into two edges $(v_1,v')$ and $(v',v_2)$,
    with the length of $(v_1,v')$ chosen such that
    the length of the subpath of $P_u$ from $u$ to $v'$ becomes exactly $\tau_i$.
    We can now select this subpath for $u$,
    as it lies inside $B_u$
    (and thus will be disjoint from any other subpath).


    
\end{proof}

\begin{lemma}\label{lem:upper-bound-top}
    We have $\opt \leq 6 \cdot \sum_{i=0}^L M_i\cdot \tau_i$.
\end{lemma}

Before we prove \cref{lem:upper-bound-top}, let us first use it to finish the proof of \cref{thm:steiner-forest-main}. 

\begin{proof}[Proof of \cref{thm:steiner-forest-main}]
As for the approximation factor, note that by \cref{lem:upper-bound-top} we have 
\[\opt\leq 6\cdot \sum_i M_i\cdot \tau_i\leq 6\cdot  \sum_i \hat{M}_i\cdot \tau_i = 6\cdot\sol . \]
Also, by \cref{lem:lower-bound-opt} we have 
\[\sol = \sum_i \hat{M}_i \cdot \tau_i \leq \cmis \cdot \sum_i M_i\cdot \tau_i \leq \cmis \cdot (L+1) \cdot (\max_i M_i\cdot \tau_i) \leq  \cmis \cdot (L+1) \cdot \opt \leq O(\log k) \cdot \opt . \]
Thus, our solution is an $O(\log k)$-approximation of $\opt$.

As for the runtime, our algorithm estimates the MIS size for each of the $L+1 = O(\log k)$ threshold values using the algorithm of \cref{thm:mis-main}, and thus has a total runtime of $\tO(n^{3/2})$. Note that the set of vertices $V_i^H$ can be computed using $O(k)$ queries. Further, checking whether two vertices have an edge in $E_i^H$ can be done by one query to $w$, and thus we have an adjacency matrix oracle access for $H_i$. 
\end{proof}

The rest of this section is devoted to the proof of~\cref{lem:upper-bound-top}. We start by introducing further notation.

\paragraph{Clusters and centers.} Let $U_i\subseteq V_i^H$ denote the MIS whose size is $M_i$. For each active vertex $v\in V_i^H$, we define its {\em center} $c_i(v)=\argmin_{u\in U_i} w(v,u)$ to be the closest point in the MIS $U_i$ to $v$, breaking ties arbitrarily.
Note that by maximality of $U_i$ we have $w(v,c_i(v)) < 2 \tau_i$.
Furthermore, for each $u\in U_i$ we define the {\em cluster} of $u$, i.e.,  $C_i(u) =\{v\in V_i^H\colon c_i(v)=u\}$ as the set of active terminals whose center is $u$ at level $i$. Finally, we use $\cC_i=\{C_i(u)\colon u\in U_i\}$ to denote the set of all clusters of level $i$.

First we show that we can connect all clusters internally by paying at most $O(\sum_i M_i\cdot \tau_i)$.

\begin{lemma}[Connectivity within clusters]\label{lem:connect-clusters}
    There exists a subset of edges $F\subseteq V\times V$ of total cost at most $4\cdot \sum_i M_i\cdot \tau_i$ such that any pair of vertices $v_1,v_2$ inside each cluster in $\cC_0\cup\cdots\cup\cC_L$ is connected in the graph $G(V,F)$.
\end{lemma}
\begin{proof}
We let $F_0=\emptyset$.
We will construct sets of edges $F_1,\dots,F_L$ such that
\begin{itemize}
    \item{(Cost Constraint)} for each $i\in \{1,\dots, L\}$, the cost of $F_i$, defined as $\sum_{(x,y)\in F_i} w(x,y)$, is at most $4M_i\tau_i$,
\end{itemize} 
while inductively maintaining the following property:
\begin{itemize}
    \item (Connectivity Property) for  $0\leq i\leq L$, any pair of vertices $v_1,v_2$ both belonging to any cluster in $\cC_0\cup \cdots\cup\cC_i$ is connected in the graph $G(V,F_0\cup\cdots\cup F_i)$.
\end{itemize}

Setting $F=\bigcup_{i=0}^L F_i$, the above cost constraint and connectivity property imply the statement of the lemma.

\medskip
As the base case,
let us ensure the connectivity property for $i=0$. Note that, by the preprocessing stage, we know that the minimum pairwise distance between any two active terminals is at least $2$. Given that $\tau_0=1$, the graph $H_0$ has no edges, and thus the only MIS in this graph contains all vertices in $V_0^H$. Therefore, all clusters in $\cC_0$ have size $1$ and thus the property holds trivially for $i=0$.
\medskip

Next, suppose that the connectivity property holds up to level $i = \ell$. This in particular means that all vertices inside each cluster of $\cC_{\ell}$ are already connected in $G(V,F_0\cup\cdots\cup F_{\ell})$. We will now show how to pick $F_{\ell+1}$ with cost at most $4M_{\ell}\cdot \tau_{\ell}$, satisfying the cost constraint, such that the vertices inside each cluster of $\cC_{\ell+1}$ will become connected in $G(V,F_0\cup\cdots\cup F_{\ell+1})$.
\medskip

Consider an undirected simple graph $G'$ where we merge each cluster of $\cC_{\ell}$ into a supernode. More precisely, for each vertex in $U_{\ell}$, there is a vertex in $G'$.
Now, we connect two supernodes corresponding to $u_1,u_2\in U_{\ell}$ with an edge in $G'$ if there exists a vertex $v\in C_{\ell}(u_1)$ and a vertex $u\in C_{\ell}(u_2)$ such that $c_{\ell+1}(v)=u$;
i.e., on the next level, the center of some point $v$ in one cluster $C_{\ell}(u_1)$ will be a point $u$ in the other cluster $C_{\ell}(u_2)$.

For the connectivity property to hold for level $\ell+1$, we need to make sure that $v$ and $u$ get connected after adding $F_{\ell+1}$.
Given that both clusters $\cC_{\ell}(u_1)$ and $\cC_{\ell}(u_2)$ are already internally connected in $G(V,F_0\cup\cdots\cup F_{\ell})$, 
to connect $u$ and $v$ it will be enough to connect $\cC_{\ell}(u_1)$ to $\cC_{\ell}(u_2)$ in any way.
Intuitively, an edge in $G'$ means that we need to connect the supernodes corresponding to its endpoints (possibly indirectly).

Note that because $u$ belongs to the MIS $U_{\ell+1}$ and serves as the center for $v$, i.e., $c_{\ell+1}(v)=u$,
we have that $u$ and $v$ are connected in $H_{\ell+1}$, and thus their distance is at most $2\tau_{\ell+1}$.
This shows that if $(u_1,u_2)$ is an edge in $G'$,
then the distance between $\cC_{\ell}(u_1)$ and $\cC_{\ell}(u_2)$ is at most $2\tau_{\ell+1}$.


\medskip
Now we state how to construct $F_{\ell+1}$:
We pick a spanning forest in $G'$ and for each edge $(u_1,u_2)$ in the spanning forest, we add the closest pair of vertices in the clusters of the corresponding supernodes, i.e., $C_{\ell}(u_1)$ and $C_{\ell}(u_2)$, to $F_{\ell+1}$.

First, note that if there is an edge between two supernodes in $G'$, then they are in the same connected component of $G'$, and because we pick a spanning forest, the two supernodes will become connected by $F_{\ell+1}$. Thus all connectivity requirements for the clusters in $\cC_{\ell+1}$ will be satisfied and we will have the connectivity property for level $\ell+1$. Second, the total number of edges we pick in the spanning forest is at most $|U_{\ell}|-1<M_{\ell}$, and the cost of each edge added to $F_{\ell+1}$ is at most $2\tau_{\ell+1}\leq 4\tau_{\ell}$. Thus, the total cost of $F_{\ell+1}$ is at most $4M_{\ell}\tau_{\ell}$, satisfying the cost constraint.
\end{proof}

Next, we need to connect each terminal $s\in T$ to its match, $m(s)$, using the available budget $O(\sum_i M_i\cdot \tau_i)$. Again, we define further notation.

\paragraph{Target terminals.} Consider a terminal $s\in T$ such that $s\in V_i^H$ but $s\notin V_{i+1}^H$. The goal is to connect all such terminals $s$ to their matches $m(s)$ using the budget of level $i$, i.e., $O(M_i\cdot \tau_i)$. Thus we define the set of {\em Target} terminals as $\target_i=V_i^H\setminus V_{i+1}^H$. Further, for $u\in U_i$ in the MIS, let $\target_i(u)=\target_i\cap C_i(u)$ be the target set inside the cluster of $u$.

\medskip

\begin{lemma}[Connecting target pairs]\label{lem:target-connect}
    For every $i=0,...,L$,
    there exists a subset of edges $J_i\subseteq V\times V$ of total cost at most $2M_i\cdot \tau_i$ such that all terminals in $\target_i$ are connected to their match in the graph $G(V,F\cup J_i)$. Here $F\subseteq V\times V$ is the subset of edges given by \cref{lem:connect-clusters}. 
\end{lemma}
\begin{proof}
   
We will again define an undirected simple graph $G'$ where we merge each cluster of $\cC_i$ into a supernode. 
Note that because we have included the set $F$ in $G(V,F\cup J_i)$,  all the nodes within a supernode are already connected.

The vertex set of $G'$ is defined similarly as in the proof of \cref{lem:connect-clusters}, where for each vertex in $U_{i}$, there is a vertex in $G'$. 
However the edges are now defined differently. Two supernodes corresponding to $u_1,u_2\in U_{i}$ are connected in $G'$ if there exists a vertex $s\in \target_i(u_1)$ such that its match is in the other cluster, i.e., $m(s)\in C_{i}(u_2)$.

Our goal is to pick a set of edges $J_i$ such that the pair $(s,m(s))$ gets connected. Given that the clusters are connected internally by \cref{lem:connect-clusters} using $F$, it is enough to ensure that we connect the supernodes corresponding to $u_1$ and $u_2$ (possibly indirectly) whenever there is an edge between them in $G'$.

Further, note that picking an edge in $G'$ can be implemented by adding an edge to $J_i$ with a cost of at most $2\tau_{i}$. This is because $s$ and $m(s)$ are in $\target_i$ and thus are not in $V_{i+1}^H$. This means that their distance is at most $\tau_{i+1}\leq 2\tau_i$. 

\medskip
The rest of the proof follows that of \cref{lem:connect-clusters}.
We state how to construct  $J_{i}$: Again  we pick a spanning forest in $G'$ and for each edge $(u_1,u_2)$ in the spanning forest, we add the closest pair of vertices in the clusters of the corresponding supernodes, i.e., $C_{i}(u_1)$ and $C_{i}(u_2)$, to $J_{i}$.

First, note that if there is an edge between two supernodes in $G'$, then they are in the same connected component of $G'$, and because we pick a spanning forest, the two supernodes will get connected by adding $J_i$. Thus all connectivity requirements for the pairs in $\target_i$ are satisfied. Second, the total number of edges we pick in the spanning forest is at most $|U_i|-1<M_i$, and the cost of each edge is at most $2\tau_i$. Thus the total cost of $J_i$ is bounded by $2M_i\tau_i$.
\end{proof}

\begin{proof}[Proof of \cref{lem:upper-bound-top}]
Note that for each terminal pair $(s_j,t_j)$ there exists an index $i\leq L$ for which $s_j\in \target_i$. Therefore, if we use \cref{lem:target-connect} for all values of $0\leq i\leq L$, then all terminals in $T$ will be connected to their match in the graph $G(V,F\cup J_0\cup \cdots \cup J_L)$.
Thus, $OPT$ is at most the cost of $F\cup J_0\cup \cdots \cup J_L$, which by \cref{lem:connect-clusters} and \cref{lem:target-connect} is at most $6\cdot \sum_i M_i\cdot\tau_i$, as desired.
\end{proof}

\subsection{Tightness of the MIS Approach for Steiner Forest}
Here, we show that the $O(\log k)$-approximation analysis of our MIS-based estimate for Steiner Forest is tight.

\begin{lemma}
    The estimate $\sum_{i} M_i \tau_i$ cannot get better than $O(\log k)$-approximation for the Steiner Forest estimation problem.
\end{lemma}
\begin{proof}
We provide two instances $I_1, I_2$, such that on $I_1$, the value of the estimate is of $O(\opt(I_1))$ and on $I_2$, the value of estimate is of $O(\opt(I_2) \cdot \log k)$.
Since we require the estimate to always be greater than the optimal Steiner Forest cost, our MIS-based estimate, possibly after a fixed scaling, cannot guarantee better than $O(\log k)$-approximation.
In both instances, there are $n = 2^L -1$ vertices $v_1, \dots, v_n$ and the number of terminal pairs $k$ is $O(n)$.

\paragraph{Construction of $I_1$.}
In this instance, these $n$ vertices are located on a line (i.e., one-dimensional Euclidean space) where the distance of any consecutive pair of vertices is exactly $2$. Formally, for every $i,j\le n$, $d(v_i, v_j) = 2(i-j)$. Finally, the terminal pairs are $\{(v_i, v_{i+n/2})\}_{i\le n/2}$. 

\begin{figure}[!h]
    \centering    \includegraphics[width=0.8\linewidth]{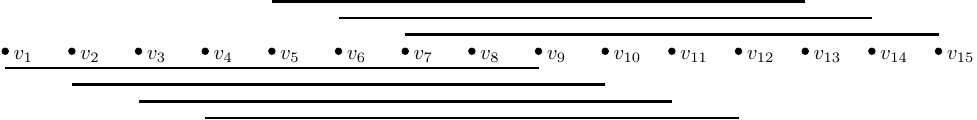}
    \caption{An example of $I_1$ with $15$ vertices. The endpoints of line segments denote terminal pairs.}
    \label{fig:I1}
\end{figure}

For any level $j< L$ in the MIS based algorithm, all pairs are active and it is straightforward to check that the size of every MIS is $O(n/ 2^j)$. Then, the estimate will have cost at least $\sum_{j < L} 2^j \cdot O(n/2^j) = O(n L) = O(n \cdot \log k)$. However, the optimal Steiner Forest on this instance has cost $O(n)$.

\paragraph{Construction of $I_2$.} In this instance, the vertices are partitioned into $L$ disjoint clusters $\mathcal{V}_1, \dots, \mathcal{V}_L$ such that for every $i\le L$, the cluster $\mathcal{V}_i$ contains $n/2^i$ vertices $v^i_1, \dots, v^i_{n/2^i}$. In each cluster $\mathcal{V}_i$, vertices are located on a line where the distance of every consecutive pair of vertices is $2^{i+1}$. Moreover, for every $i,j\le L$, the minimum distance between any two clusters $\mathcal{V}_i$ and $\mathcal{V}_j$ is $M \gg n$. Finally, the terminals are of the form $\{(v^i_1, v^i_2), \ldots (v_{n/2^i - 1}^i, v^i_{n/2^i})\}_{i\le L}$.

\begin{figure}[!h]
    \centering    \includegraphics[width=0.5\linewidth]{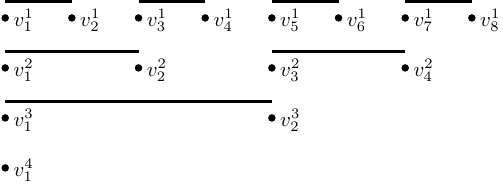}
    \caption{An example of $I_2$ with $15$ vertices. The endpoints of line segments denote terminal pairs.} 
    \label{fig:I2}
\end{figure}

For any level $j< L$ in the MIS based algorithm, all pairs in clusters $\mathcal{V}_{j}, \ldots, \mathcal{V}_{L}$ are active and it is straightforward to check that the size of every MIS is $O(n/ 2^j)$. Then, the estimate will have cost $\sum_{j < L} 2^j \cdot O(n/2^j) = O(n L) = O(n \cdot \log k)$. Moreover, the optimal Steiner Forest cost is $\sum_{j \le L} 2^j \cdot n/2^j = O(nL) = O(n \cdot \log k)$.
\end{proof}
\section{Sublinear Time MIS under the Adjacency Matrix Model}\label{sec:MIS}

In this section, we design a sublinear time algorithm for estimating the MIS size to a multiplicative $1+\eps$ approximation factor.

\subsection{Preliminaries}

We begin by stating the basic definitions and facts needed in this section.
We denote the input graph by $(V,M)$ to emphasize that the edge set is accessed via an adjacency matrix oracle, and
we set $n = |V|$.

\begin{definition}
\textnormal{(Permutations)}
    \begin{itemize}
        \item For a set $V$ we denote by $S(V)$ the set of permutations of $V$, i.e., bijections from $\{1,...,|V|\}$ to $V$.
        \item For a subset $V' \subseteq V$ and a permutation $\pi \in S(V)$
we define the \emph{restriction of $\pi$ to $V'$}, denoted by $\pi[V']$,
as the permutation of $V'$ obtained from $V$ by keeping only those elements that belong to $V'$ and renumbering appropriately
(i.e., $\pi[V'](i)$ is the $i$-th element of $V'$ among $\pi(1), \pi(2), ...$).
    \end{itemize}
\end{definition}

\begin{proposition}[Chernoff Bound]\label{prop:chernoff}
    Consider independent Bernoulli random variables $X_1, X_2, \dots, X_r$, and define their sum as  $X = \sum_{i=1}^{r} X_i$. For any $k > 0$, the probability that $X$ deviates from its expectation can be bounded by 
    \[
        \Pr\big(|X - \mathbb{E}[X]| \geq k\big) \leq 2 \exp \left(-\frac{k^2}{3\mathbb{E}[X]}\right).
    \]
\end{proposition}


\begin{definition} \label{def:mult-add-apx}
For $\alpha \geq 1$, $\beta \ge 0$, we say that $\widetilde{x}$ is an $(\alpha,\beta)$-estimate of a true value $x$ if and only if  $x \leq \widetilde{x} \leq \alpha x + \beta$. We say that an algorithm is an $(\alpha,\beta)$ [multiplicative-additive] approximation if it computes such an estimate of the relevant true value.
\end{definition}

\begin{definition}
    A maximal independent set (MIS) is a subset $V' \subseteq V$ of vertices such that
    for any edge $(u,v) \in M$, at most one of $u$, $v$ is in $V'$, and for any $v \in V \setminus V'$, some neighbor of $v$ is in $V'$. 
\end{definition}

\begin{definition}
    Given a permutation $\pi \in S(V)$, start with an empty set $I := \emptyset$, and for $i=1,...,n$, add vertex $\pi(i)$ to $I$ if none of its neighbors are in $I$.
    The set $\RGMIS(\pi)$ is defined as the resulting set $I$.
\end{definition}

\begin{fact}
    For any $\pi \in S(V)$, the set $\RGMIS(\pi)$ is an MIS.
    Moreover, for any $v \in V$,
    we have that
    $v \in \RGMIS(\pi)$ if any only if for every neighbor $u$ of $v$ of lower rank in $\pi$ (i.e., $\pi^{-1}(u) < \pi^{-1}(v)$) we have $u \not \in \RGMIS(\pi)$.
\end{fact}

The following oracle $\cA$,
proposed by Nguyen and Onak~\cite{NguyenOnakFOCS08}
and analyzed by Yoshida, Yamamoto and Ito~\cite{YoshidaYISTOC09},
answers the query of whether a vertex $v$ is in $\RGMIS(\pi)$.
We call it an ``abstract" oracle since it does not specify how to implement Line~1,
and we will only use the bound on its query complexity (i.e., number of recursive calls to $\cA$) given by \cref{thm:yyi}.

\begin{algorithm}[H]
\caption{$\cA(\pi,v)$ (abstract $\RGMIS$ oracle~\cite{NguyenOnakFOCS08,YoshidaYISTOC09})}
\label{alg:abstract-oracle}

let $u_1, u_2, ..., u_k$ be all neighbors of $v$, sorted in increasing order of their ranks in $\pi$

\For{$i=1,...,k$}{
    \If{\textnormal{$u_i$ has lower rank than $v$ in $\pi$}}{
        \If{$\cA(\pi,u_i)$}{
            \Return \textbf{false}
        }
    }
}

\Return \textbf{true}
\end{algorithm}

\begin{theorem}[\cite{YoshidaYISTOC09}] \label{thm:yyi}
    Let $T_\cA(\pi,v)$ be the total number of recursive calls to $\cA$ for a query $\cA(\pi,v)$.
    Then we have
    \[
        \E_{v \in V, \pi \in S(V)} \left[T_\cA(\pi,v)\right] \le 1 + \frac{|M|}{|V|} .
    \]
\end{theorem}

We remark that the right-hand side is at most $O(\Delta)$,
and that the bound is in expectation over both $\pi$ and $v$.


\subsection{Single-sample algorithm} \label{sec:single-sample-mis}

As discussed in the introduction, a straightforward implementation of Line~1 in \cref{alg:abstract-oracle} would result in a running time of $O(\Delta n)$.
In this subsection we describe our more efficient implementation of the vertex oracle in the adjacency matrix model, \cref{alg:O}.
We note that if we are satisfied with a running time of $O(n)$,
we can materialize the entire random permutation $\pi$.
Then, when the oracle is called for a vertex $v$ and needs to iterate over its neighbors in increasing order of rank,
we loop over vertices $\pi(1)$, $\pi(2)$, $\pi(3)$, ..., and check whether each of them is a neighbor of $v$. 

\begin{algorithm}[H]
\caption{$\cO(\pi,v)$ (our MIS oracle for the adjacency matrix model)}
\label{alg:O}
let $k = \pi^{-1}(v)$ be the rank of $v$ in $\pi$

\For{$i=1,...,k-1$}{
    \If{$(\pi(i),v) \in M$}{
        \If{$\cO(\pi,\pi(i))$}{
            \Return \textbf{false}
        }
    }
}

\Return \textbf{true}
\end{algorithm}

In terms of the recursive calls being made and the results returned,
the oracles $\cO$ and $\cA$ (\cref{alg:abstract-oracle,alg:O}) are clearly equivalent.
However, one might worry that $\cO$ (\cref{alg:O}) potentially spends a lot of time querying non-edges on Line~3.

Intuitively, for a vertex $u$ of degree $\deg(u)$, it takes expected $n/\deg(u)$ time to find a neighbor of $u$ if we queried the adjacency matrix at random.
Vertex degrees in a graph can vary widely, but this is not an issue:
the proof of \cref{thm:yyi} in~\cite{YoshidaYISTOC09}
in fact shows that for any vertex $u$,
the expected number of recursive calls out of $\cA(\pi,u)$
is at most $\deg(u)/n$
for a random ``top-level" query vertex~$v$.
Thus, one could hope that if finding one neighbor of $u$ to call recursively took $n/\deg(u)$ time (i.e., queries to the adjacency matrix),
then the total runtime would be given by a calculation such as $\sum_{u \in V} \frac{\deg(u)}{n} \cdot \frac{n}{\deg(u)} = O(n)$.

However, such an argument is invalidated by the fact that
both the trajectory of the recursive oracle calls (to $\cO$)
and the sequence of adjacency matrix calls (to $M$, in order to find neighbors)
are functions of the random permutation $\pi$,
and thus seem highly correlated.\footnote{%
    An argument as outlined above was recently used
    by~\cite{mahabadi2025sublinear} in the context of maximal matchings.
    However, their algorithm uses fresh, independent randomness to query random entries of the adjacency matrix to find neighbors of the current vertex,
    which enables the argument.
}
In fact, every vertex~$u$ queries the same vertices $\pi(1)$, $\pi(2)$, ... to find its neighbors.
Nevertheless, we are still able to prove a running time bound of $O(n)$,
albeit by employing a different argument in the proof.

\begin{lemma} \label{lem:O-runtime}
    Let $T_\cO(\pi,v)$ be the total runtime of \cref{alg:O} (including all recursive calls) for query $\cO(\pi,v)$.
    Then we have
    \[
        \E_{v \in V, \pi \in S(V)} \left[T_\cO(\pi,v)\right] = O(n).
    \]
\end{lemma}

We remark that if one is not worried about logarithmic factors,
then the strategy above can still be executed to obtain a $\tO(n)$ bound,
which one could argue is simpler than our proof below,
with the caveat that
this would also require modifying \cref{alg:O} to introduce caching of the returned answers.
See \cref{remark:simpler-proof} for a further discussion.

The rest of
this section
is devoted to the proof of \cref{lem:O-runtime}.
Our plan is to couple the execution of our \cref{alg:O} ($\cO$)
on the original input graph $(V,M)$
to the execution of \cref{alg:abstract-oracle} ($\cA$)
on a certain larger graph $H$.

We start by defining the latter.
Informally, we start with a copy of $(V,M)$ as a first layer.
As the second layer $V_2$, we adjoin $2n$ new vertices, with no edges inside.
As the third layer $V_3$, we adjoin $n$ new vertices, with no edges inside.
Finally, we connect each vertex in the second layer to each vertex in the first and the third layers.
Formally:

\begin{definition}
Let $H$ be a graph with vertex set $V(H) = V \cup V_2 \cup V_3$,
where $V$ is the vertex set of $(V,M)$,
$|V_2|=2n$, and $|V_3|=n$.
We define the edge set $E(H) = M \cup (V \times V_2) \cup (V_2 \times V_3)$.
\end{definition}

Next, we identify requirements on the permutation $\pi_H$ of vertices of $H$
that will enable our coupling argument.

\begin{definition} \label{def:good-pi}
    We say that a permutation $\pi_H \in S(V(H))$ is \emph{good}
    if: 
    \begin{itemize}
        \item the first vertex comes from $V_3$, i.e., $\pi_H(1) \in V_3$,
        \item in every prefix of $\pi_H$ there are at least as many vertices from $V_2$ as from $V$.
    \end{itemize}
\end{definition}

We show that a random permutation $\pi_H$ is indeed likely to be good:
\begin{lemma} \label{lem:good-perm}
    Consider a uniformly random permutation $\pi_H \in S(V(H))$.
    Then:
    \begin{enumerate}
        \item $\Prob[\text{$\pi_H$ is good}] \ge \frac{1}{12}$.
        \item If $\pi_H$ is a uniformly random {good} permutation, $\pi := \pi_H[V]$ ($\pi_H$ restricted to $V$) is uniformly random in $S(V)$.
    \end{enumerate}
\end{lemma}
\begin{proof}
    Note that $\pi_H(1) \in V_3$ with probability $\frac{|V_3|}{|V(H)|} = \frac{1}{4}$,
    and that conditioning on this event reveals no information about the relative order of vertices in $V \cup V_2$,
    so we can focus on the restricted permutation $\pi_H[V \cup V_2]$.
    We now employ a statement known as Bertrand's ballot theorem~\cite{whitworth,bertrand1887solution}:
    in an election where candidate A receives $p$ votes and candidate B receives $q$ votes
    with $p>q$,
    if the votes are counted in random order,
    the probability that A will be strictly ahead of B throughout the count is $\frac{p-q}{p+q}$.
    This maps to our scenario by setting $p = |V_2| = 2n$ and $q = |V| = n$,
    resulting in a probability of $\frac{2n-n}{2n+n} = \frac{1}{3}$.

    For  the second part,
    note that conditioning on both events reveals no information about the relative order of vertices in $V$.
\end{proof}

Now assume that $\pi_H \in S(V(H))$ is good,
and consider what happens in the $\RGMIS(\pi_H)$ process in $H$.
First, $\pi_H(1) \in V_3$ joins the MIS, and thus removes all $V_2$ vertices from consideration
(as it is connected to all of them).
Other $V_3$ vertices will join the MIS when they arrive.
All the remaining action happens on $(V,M)$ according to the permutation $\pi = \pi_H[V]$
($\pi$ restricted to $V$).
That is, we have $\RGMIS(\pi_H) = V_3 \cup \RGMIS(\pi)$.

Moreover, consider the tree of recursive calls made
when executing $\cO(\pi,v)$
for a query vertex $v \in V$
(in the original graph $(V,M)$),
and compare it to the tree of recursive calls made when executing $\cA(\pi_H,v)$ (in $H$).
The difference is that in the latter, there will be some \emph{additional} calls from a vertex in $V$ to $\cA(\pi_H,v_2)$ for a vertex $v_2 \in V_2$; that will then call $\cA(\pi_H,\pi_H(1))$, which returns true, so $\cA(\pi_H,v_2)$ returns false.
Otherwise, the two trees are the same.
In particular, the query complexity of $\cO(\pi,v)$
(the number of recursive calls made)
is at most that of $\cA(\pi_H,v)$.

However, $\cO$ also makes queries to $M$ (particularly for non-edges), and we aim to upper-bound their number.
We will charge this to the aforementioned \emph{additional} recursive calls of $\cA$;
to that end, we will use the second condition of \cref{def:good-pi}.

Consider an execution of $\cO(\pi,u)$ for any $u \in V$
(without counting its recursive calls).
In the no-case
(if $\cO(\pi,u)$ returns false)
it queries $M$ for pairs $(\pi(i),u)$ for $i=1,2,...,\ell$ for some $\ell$,
and $\cO(\pi,\pi(\ell))$ is the first direct recursive call that returns true.
In the yes-case (if $\cO(\pi,u)$ returns true) it queries
$M$ for pairs $(\pi(i),u)$ for $i=1,2,...,k-1$,
where $k=\pi^{-1}(u)$ is the rank of $u$ in~$\pi$;
denote $\ell := k-1$ in the yes-case.
Denote by $L$ the rank of $\pi(\ell)$
in $\pi_H$
(i.e., $\pi_H(L) = \pi(\ell)$).
Then, we can say that $\cA(\pi_H,u)$ must have made recursive calls to all neighbors of $u$ in $H$ with rank at most $L$ in $\pi_H$.
(In the yes-case, $\cA$ queries all neighbors of lower rank before returning true;
in the no-case, $\cA(\pi_H, \pi_H(L)) = \cA(\pi_H, \pi(\ell))$ is the first recursive call that returns true.)
These neighbors in particular include all vertices in $V_2$ of at most that rank,
because $u$ is connected to all $V_2$ vertices in $H$.
Thus,
the number of calls to $M$ made directly
(without counting recursive calls)
by an execution of
$\cO(\pi,u)$ for any $u \in V$
is
\begin{align*}
    \ell &= |\{\pi_H(1), ..., \pi_H(L)\} \cap V| \\
         &\le |\{\pi_H(1), ..., \pi_H(L)\} \cap V_2| \\
         &\le \text{number of \emph{additional} recursive calls to $V_2$ made directly by $\cA(\pi_H,u)$},
\end{align*}
where the first inequality follows since $\pi_H$ is good.
Now, summing this up over the entire execution tree of $\cO(\pi,v)$, we get that the total number of calls to $M$
is at most the number the \emph{additional} recursive calls to $V_2$ made by $\cA$.
Together with the above observation that the trees of $\cO(\pi,v)$ and $\cA(\pi_H,v)$ are the same when restricted to vertices in $V$, we have shown:
\begin{lemma} \label{lem:O-by-A}
    For any good $\pi_H \in S(V(H))$ and any $v \in V$,
    the running time of $\cO(\pi,v)$
    is upper-bounded by
    the oracle complexity of $\cA(\pi_H,v)$
    (up to a constant factor),
    i.e.,
    \[
        T_\cO(\pi,v) \le O(T_\cA(\pi_H,v)) ,
    \]
    where $\pi = \pi_H[V]$.
    \qed
\end{lemma}
We note that \cref{thm:yyi} allows us to bound the expectation of $T_\cA$
over a random permutation $\pi_H \in S(V(H))$ (which may not be good)
and a random query vertex $v \in V(H)$ (which may not be in $V$).
However, $V$ constitutes a constant fraction of $V(H)$, and a constant fraction of permutations are good;
therefore, if the expectation over good permutations and over vertices in~$V$ was large,
the entire expectation would also be large (up to a constant factor).
Thus we have:
\begin{align*}
    \E_{v \in V, \; \pi \in S(V)}[T_\cO(\pi,v)]
    &= \E_{v \in V, \; \text{good} \; \pi_H \in S(V(H))}[T_\cO(\pi_H[V],v)] & \text{(by \cref{lem:good-perm}.2)}\\
    &\le O(1) \cdot \E_{v \in V, \; \text{good} \; \pi_H \in S(V(H))}[T_\cA(\pi_H,v)] & \text{(by \cref{lem:O-by-A})} \\
    &\le O(1) \cdot \E_{v \in V(H), \; \pi_H \in S(V(H))}[T_\cA(\pi_H,v)] & \text{(by \cref{lem:good-perm}.1)}\\
    &\le O\left(\frac{|E(H)|}{|V(H)|}\right) & \text{(by \cref{thm:yyi})} \\
    &= O(n).
\end{align*}
This concludes the proof of \cref{lem:O-runtime}.
\qed

\begin{remark} \label{remark:simpler-proof}
    We note that for a version of \cref{alg:O} where we introduce caching (storing and reusing the results of each recursive call), there is a different and slightly simpler proof of \Cref{lem:O-runtime} (ignoring $\log n$ factors in the running time),
    which was kindly pointed out to us by an anonymous reviewer. Specifically, in $\cA(\pi,u)$,
    one can ensure that it takes time $O(n/\deg(u) \cdot \log n)$
    to find each neighbor of $u$,
    with high probability for most permutations, via an application of the Chernoff bound. This is because if we consider a random permutation of vertex $u$'s row in the adjacency matrix, then, with high probability, any $\Theta(n/\deg(u) \cdot \log n)$ consecutive entries will contain at least one entry equal to 1.
    For permutations that do not satisfy this, we use a simple running time upper bound of $O(n^2)$ that holds when caching is used.
    
    We prefer
    to give a tight (up to $O(1)$ factors) analysis for the $\cO$-oracle (tight as it clearly runs in $\Omega(n)$ time). This is in the spirit of the analysis of \cite{YoshidaYISTOC09}, whose bound of $1 + |E|/|V|$ (\cref{thm:yyi}) is exactly tight and does not rely on caching.
\end{remark}

\subsection{Algorithm to estimate the MIS size}

\cref{alg:O} by itself only gives an answer for a single vertex.
To obtain an additive error of $\tO(n/s)$, we iterate it by sampling $\tO(s)$ vertices.
This routine, \cref{alg:addmul}, is one of the two major building blocks of our final algorithm for MIS.

\begin{algorithm}[H]
\caption{$\AlgAddMul(V,M,s,\pi)$}
\label{alg:addmul}
let $r := \frac{27s}{\eps^2} \log(n)$ 

\For{$i=1,...,r$}{
    sample a random vertex $v$ from $V$

    let $X_i := \cO(\pi,v)$ \Comment{invoke \Cref{alg:O}}
}

\Return $(1+\frac{\eps}{2})\cdot\left(\frac{\sum_i X_i}{r} \cdot n + \frac{\eps n}{3s}\right)$
\end{algorithm}

We remark that we add the $\Theta(\eps)$ terms to obtain single-sided additive error (see \cref{def:mult-add-apx}).

\begin{lemma} \label{lem:mis-add-mul}
    For any $\eps \in (0,1)$ and $s \ge 1$,
    \cref{alg:addmul} ($\AlgAddMul$),
    given a graph $(V,M)$
    with oracle access to its adjacency matrix,
    as well as a uniformly random permutation $\pi$ of its vertices,
    with probability $1-1/\poly(n)$
    reports a $(1 + \eps,\eps n / s)$-multiplicative-additive approximation
    to the value $|\RGMIS(\pi)|$
    and runs in expected $\tO(ns/\eps^2)$ time.
\end{lemma}
\begin{proof}
    First, we prove that the output of $\AlgAddMul$ is a $(1 + \eps ,\eps n / s)$-multiplicative-additive approximation to the value of $|\RGMIS(\pi)|$. Since $v$ is chosen uniformly at random in Line 3 of $\AlgAddMul$, we have
    \begin{align*}
        \E[X_i] = \Pr[X_i = 1] = \frac{|\RGMIS(\pi)|}{n}.
    \end{align*}
    Let $X = \sum_{i=1}^r X_i$. Then,
    \begin{align*}
        \E[X] = \frac{r\cdot |\RGMIS(\pi)|}{n}.
    \end{align*}
    Since $X$ is the sum of $r$ independent Bernoulli random variables, we can apply the Chernoff bound (\Cref{prop:chernoff}), which implies
    \begin{align*}
        \Pr\big(|X - \mathbb{E}[X]| \geq \sqrt{3\E[X] \log n}\big) \leq 2 \exp \left(-\frac{ 3\E[X] \log n}{3\mathbb{E}[X]}\right) \leq \frac{2}{n}.
    \end{align*}
    Thus, with probability of $1 - 1/\poly(n)$, it holds that
    \begin{align*}
        \frac{nX}{r} &\in \frac{n(\E[X] \pm \sqrt{3\E[X] \log n})}{r} \\
        &= \frac{n\E[X]}{r} \pm \frac{\sqrt{3n^2 \E[X] \log n}}{r}\\
        & = |\RGMIS(\pi)| \pm \sqrt{\frac{3n|\RGMIS(\pi)| \log n}{r}} & (\text{since } \E[X] = \frac{r\cdot |\RGMIS(\pi)|}{n})\\
        & = |\RGMIS(\pi)| \pm \sqrt{\frac{\eps^2 n|\RGMIS(\pi)|}{9s}} & (\text{since } r = \frac{27s}{\eps^2} \log(n)).
    \end{align*}
    Let $\vartheta$ be the output of the algorithm. We consider two possible scenarios:
    \begin{itemize}
        \item $|\RGMIS(\pi)| \leq n/s$: In this case, we have
        \begin{align*}
            \vartheta = \left(1+\frac{\eps}{2} \right) \cdot \left( \frac{nX}{r} + \frac{\eps n}{3s}\right) &\in \left(1+\frac{\eps}{2} \right) \cdot \left(|\RGMIS(\pi)| + \frac{\eps n}{3s} \pm \sqrt{\frac{\eps^2 n|\RGMIS(\pi)|}{9s}} \right)\\
            & = \left(1+\frac{\eps}{2} \right) \cdot \left(|\RGMIS(\pi)| + \frac{\eps n}{3s} \pm \frac{\eps n}{3s}\right),
        \end{align*}
        where the last inequality follows by the assumption $|\RGMIS(\pi)| \leq n/s$. Thus,
        \begin{align*}
            |\RGMIS(\pi)| \leq \vartheta \leq (1+\eps)\cdot |\RGMIS(\pi)| + \frac{\eps n}{s}.
        \end{align*}
        \item $|\RGMIS(\pi)| > n/s$: In this case, we have
        \begin{align*}
           \vartheta = \left(1+\frac{\eps}{2} \right) \cdot \left( \frac{nX}{r} + \frac{\eps n}{3s}\right) &\in \left(1+\frac{\eps}{2} \right) \cdot \left(|\RGMIS(\pi)| + \frac{\eps n}{3s} \pm \sqrt{\frac{\eps^2 n|\RGMIS(\pi)|}{9s}} \right)\\
            & = \left(1+\frac{\eps}{2} \right) \cdot \left(|\RGMIS(\pi)| + \frac{\eps n}{3s} \pm \frac{\eps}{3} |\RGMIS(\pi)|\right)
        \end{align*}
        where we have the last inequality because  $|\RGMIS(\pi)| > n/s$. Hence,
        \begin{align*}
            |\RGMIS(\pi)| \leq \vartheta \leq (1+\eps)\cdot |\RGMIS(\pi)| + \frac{\eps n}{s}.
        \end{align*}
    \end{itemize}
    In both cases, the output of $\AlgAddMul$ is $(1 + \eps, \eps n /s)$-multiplicative-additive approximation to the value of $|\RGMIS(\pi)|$. Let $T(\pi)$ be the expected running time of $\AlgAddMul$ on permutation~$\pi$. Since we are sampling $r$ vertices with replacement, we have
    \begin{align*}
        T(\pi) = r \cdot \E_{v \in V}[T_\cO(\pi, v)].
    \end{align*}
    Therefore, the expected running time of the algorithm is
    \begin{align*}
        \E_\pi[T(\pi)] = \frac{\sum_\pi r \cdot \E_{v\in V}[T_\cO(\pi, v)]}{n!} &= r \cdot \E_{\pi \in S(V), v\in V}[T_\cO(\pi, v)]\\
        & = O(r\cdot n)\\
        & = \widetilde{O}(ns/\eps^2),
    \end{align*}
    which completes the proof.
\end{proof}

Now we can state our main \cref{alg:mainmis} and prove \cref{thm:mis-main}, which we restate for convenience.
Roughly, our main idea here is a meet-in-the-middle approach between estimating the MIS size using \cref{alg:addmul}, which works well if the MIS is large,
and building an MIS explicitly, which can be done if the MIS is small.

\begin{algorithm}[H]
\caption{$\AlgMul(V,M)$}
\label{alg:mainmis}
$\pi := $ uniformly random permutation of $V$

$s := \sqrt{n}$

$\CurrentMISSize := 0$

$\Active[v] := 1$ \textbf{for all} $v \in V$

$j := 0$

\While{$j < n$ \textbf{and} $\CurrentMISSize < s$}{
    $j := j + 1$

    \If{$\Active[\pi(j)] = 1$}{
        $\CurrentMISSize := \CurrentMISSize + 1$ \Comment{add $\pi(j)$ to the MIS}

        \For{$v \in V$}{
            \If{$v = \pi(j)$ \textbf{or} $(\pi(j),v) \in M$}{
                $\Active[v] := 0$
            }
        }
    }
}

\If{$\CurrentMISSize < s$ \label{line:currentmissize_s}}{
    \Return $\CurrentMISSize$ \Comment{we have explicitly built the entire MIS}
}\Else{
    $V' := \{v \in V : \Active[v] = 1\}$

    $\pi' := \pi[V']$ ($\pi$ restricted to $V'$)

    \Return $\CurrentMISSize + \AlgAddMul(V',M,s,\pi')$
}
\end{algorithm}

\mismainthm*

We remark that our algorithm generates a random permutation $\pi \in S(V)$
(or $O(\log n)$ such permutations to obtain the runtime bound with high probability)
and estimates $|\RGMIS(\pi)|$ for that permutation $\pi$;
it does not attempt to estimate $\E_\pi[|\RGMIS(\pi)|]$.

\begin{proof}
    First, we prove that the output of the algorithm is a $(1+\eps)$-approximation of the value of $|\RGMIS(\pi)|$. If the condition on \cref{line:currentmissize_s} of the algorithm holds, then the algorithm explicitly built the entire MIS, and the output is exactly equal to the size of $\RGMIS(\pi)$. Now, suppose that the condition on \cref{line:currentmissize_s} of the algorithm does not hold. Thus, $\CurrentMISSize = s = \sqrt{n}$. Note that 
    \begin{align}\label{eq:rgmis-rec}
        |\RGMIS(\pi)|= \sqrt{n} + |\RGMIS(\pi[V'])|.
    \end{align}
    Let $\vartheta$ be the output of the algorithm. Hence, we have
    \begin{align*}
        \vartheta &= \sqrt{n} + \AlgAddMul(V',M,s,\pi')\\ 
        &\geq \sqrt{n} + |\RGMIS(\pi[V'])| & (\text{by \Cref{lem:mis-add-mul}})\\
        & = |\RGMIS(\pi)| & (\text{by \Cref{eq:rgmis-rec}} ).
    \end{align*}
    On the other hand,
    \begin{align*}
        \vartheta &= \sqrt{n} + \AlgAddMul(V',M,s,\pi')\\
        & \leq \sqrt{n} + (1+\eps) \cdot |\RGMIS(\pi[V'])| + \frac{\eps n}{s} & (\text{by \Cref{lem:mis-add-mul}})\\
        & = \sqrt{n} + (1+\eps) \cdot |\RGMIS(\pi[V'])| + \eps \sqrt{n} & (\text{since } s = \sqrt{n})\\
        & = (1+\eps) \cdot \left(\sqrt{n} + |\RGMIS(\pi[V'])| \right)\\
        & = (1+\eps) \cdot |\RGMIS(\pi)| & (\text{by \Cref{eq:rgmis-rec}} ),
    \end{align*}
    which completes the proof for the approximation ratio.

   Regarding the running time, the while loop in the algorithm runs in $O(n^{3/2})$ time,
   as in each step it either skips over an inactive vertex, which happens at most $n$ times,
   or processes all vertices in $O(n)$ time; the latter happens at most
   $s = \sqrt{n}$ times. Furthermore, the relative order of the vertices in $V'$
   has no influence on the execution of the algorithm before it calls $\AlgAddMul$ on the last line;
   thus $\pi'$ is a uniformly random permutation of $V'$.
   Therefore, by \Cref{lem:mis-add-mul}, the expected running time of $\AlgAddMul$ called on the last line is $\widetilde{O}(n^{3/2}/\eps^2)$.
   
   To obtain a high-probability bound on the running time, we execute $O(\log n)$ instances of the algorithm in parallel and stop as soon as the first instance completes.
   By applying Markov’s inequality
   we conclude that each individual instance terminates within $\widetilde{O}(n^{3/2}/\eps^2)$ time with constant probability. Consequently, with high probability, at least one of these instances finishes within $\widetilde{O}(n^{3/2}/\eps^2)$ time. This completes the proof.
\end{proof}

\section*{Acknowledgment} The work was initiated while Sepideh Mahabadi and Ali Vakilian were long-term visitors at the Simons Institute for the Theory of Computing as part of the Sublinear Algorithms program.

\printbibliography

@article{ChenMetric-Arxiv22,
  author       = {Yu Chen and
                  Sanjeev Khanna and
                  Zihan Tan},
  editor       = {Kousha Etessami and
                  Uriel Feige and
                  Gabriele Puppis},
  title        = {Sublinear Algorithms and Lower Bounds for Estimating {MST} and {TSP}
                  Cost in General Metrics},
  journal    = {50th International Colloquium on Automata, Languages, and Programming,
                  {ICALP} 2023, July 10-14, 2023, Paderborn, Germany},
  series       = {LIPIcs},
  volume       = {261},
  pages        = {37:1--37:16},
  publisher    = {Schloss Dagstuhl - Leibniz-Zentrum f{\"{u}}r Informatik},
  year         = {2023},
  url          = {https://doi.org/10.4230/LIPIcs.ICALP.2023.37},
  doi          = {10.4230/LIPICS.ICALP.2023.37},
  bibsource    = {dblp computer science bibliography, https://dblp.org}
}

@inproceedings{chen2020,
  author =	{Yu Chen and Sampath Kannan and Sanjeev Khanna},
  title =	{{Sublinear Algorithms and Lower Bounds for Metric TSP Cost Estimation}},
  booktitle =	{47th International Colloquium on Automata, Languages, and Programming (ICALP 2020)},
  pages =	{30:1--30:19},
  series =	{Leibniz International Proceedings in Informatics (LIPIcs)},
  ISBN =	{978-3-95977-138-2},
  ISSN =	{1868-8969},
  year =	{2020},
  volume =	{168},
  myeditor =	{Artur Czumaj and Anuj Dawar and Emanuela Merelli},
  publisher =	{Schloss Dagstuhl--Leibniz-Zentrum f{\"u}r Informatik},
  address =	{Dagstuhl, Germany},
  myurl =		{https://drops.dagstuhl.de/opus/volltexte/2020/12437},
  URN =		{urn:nbn:de:0030-drops-124372},
  mydoi =		{10.4230/LIPIcs.ICALP.2020.30}
}

@inproceedings{BhattacharyaKS23,
  author       = {Sayan Bhattacharya and
                  Peter Kiss and
                  Thatchaphol Saranurak},
  title        = {Dynamic (1+$\epsilon$)-Approximate Matching Size in Truly Sublinear
                  Update Time},
  booktitle    = {64th {IEEE} Annual Symposium on Foundations of Computer Science, {FOCS}
                  2023, Santa Cruz, CA, USA, November 6-9, 2023},
  pages        = {1563--1588},
  publisher    = {{IEEE}},
  year         = {2023},
  url          = {https://doi.org/10.1109/FOCS57990.2023.00095},
  doi          = {10.1109/FOCS57990.2023.00095},
  bibsource    = {dblp computer science bibliography, https://dblp.org}
}

@inproceedings{kapralov2020space,
  title={Space efficient approximation to maximum matching size from uniform edge samples},
  author={Kapralov, Michael and Mitrovi{\'c}, Slobodan and Norouzi-Fard, Ashkan and Tardos, Jakab},
  booktitle={Proceedings of the Fourteenth Annual ACM-SIAM Symposium on Discrete Algorithms},
  pages={1753--1772},
  year={2020},
  organization={SIAM}
}

@article{chlebik2008steiner,
  title={The Steiner tree problem on graphs: Inapproximability results},
  author={Chleb{\'\i}k, Miroslav and Chleb{\'\i}kov{\'a}, Janka},
  journal={Theoretical Computer Science},
  volume={406},
  number={3},
  pages={207--214},
  year={2008},
  publisher={Elsevier}
}

@article{gilbert1968steiner,
  title={Steiner minimal trees},
  author={Gilbert, Edgar N and Pollak, Henry O},
  journal={SIAM Journal on Applied Mathematics},
  volume={16},
  number={1},
  pages={1--29},
  year={1968},
  publisher={SIAM}
}

@article{czumaj2009estimating,
  title={Estimating the weight of metric minimum spanning trees in sublinear time},
  author={Czumaj, Artur and Sohler, Christian},
  journal={SIAM Journal on Computing},
  volume={39},
  number={3},
  pages={904--922},
  year={2009},
  publisher={SIAM}
}

@inproceedings{chen2023query,
  title={Query Complexity of the Metric Steiner Tree Problem},
  author={Chen, Yu and Khanna, Sanjeev and Tan, Zihan},
  booktitle={Proceedings of the 2023 Annual ACM-SIAM Symposium on Discrete Algorithms (SODA)},
  pages={4893--4935},
  year={2023},
  organization={SIAM}
}

@inproceedings{Behnezhad21,
  author       = {Soheil Behnezhad},
  title        = {Time-Optimal Sublinear Algorithms for Matching and Vertex Cover},
  booktitle    = {62nd {IEEE} Annual Symposium on Foundations of Computer Science, {FOCS}},
  pages        = {873--884},
  year         = {2021}
}

@inproceedings{NguyenOnakFOCS08,
  author    = {Huy N. Nguyen and
               Krzysztof Onak},
  title     = {{Constant-Time Approximation Algorithms via Local Improvements}},
  booktitle = {49th Annual {IEEE} Symposium on Foundations of Computer Science, {FOCS}},
  pages     = {327--336},
  year      = {2008}
}

@inproceedings{OnakSODA12,
  author    = {Krzysztof Onak and
               Dana Ron and
               Michal Rosen and
               Ronitt Rubinfeld},
  title     = {{A Near-Optimal Sublinear-Time Algorithm for Approximating the Minimum
               Vertex Cover Size}},
  booktitle = {Proceedings of the Twenty-Third Annual {ACM-SIAM} Symposium on Discrete
               Algorithms, {SODA}},
  pages     = {1123--1131},
  year      = {2012}
}

@inproceedings{YoshidaYISTOC09,
  author    = {Yuichi Yoshida and
               Masaki Yamamoto and
               Hiro Ito},
  editor    = {Michael Mitzenmacher},
  title     = {An improved constant-time approximation algorithm for maximum matchings},
  booktitle = {Proceedings of the 41st Annual {ACM} Symposium on Theory of Computing,
               {STOC}},
  pages     = {225--234},
  publisher = {{ACM}},
  year      = {2009}
}

@inproceedings{BehnezhadRRS-SODA23,
  author    = {Soheil Behnezhad and
               Mohammad Roghani and
               Aviad Rubinstein and
               Amin Saberi},
  editor    = {Nikhil Bansal and
               Viswanath Nagarajan},
  title     = {Beating Greedy Matching in Sublinear Time},
  booktitle = {Proceedings of the 2023 {ACM-SIAM} Symposium on Discrete Algorithms,
               {SODA} 2023},
  pages     = {3900--3945},
  publisher = {{SIAM}},
  year      = {2023}
}

@inproceedings{TSP-icalp24,
  author       = {Soheil Behnezhad and
                  Mohammad Roghani and
                  Aviad Rubinstein and
                  Amin Saberi},
  title        = {Sublinear Algorithms for {TSP} via Path Covers},
  booktitle    = {51st International Colloquium on Automata, Languages, and Programming,
                  {ICALP}},
  series       = {LIPIcs},
  volume       = {297},
  pages        = {19:1--19:16},
  publisher    = {Schloss Dagstuhl - Leibniz-Zentrum f{\"{u}}r Informatik},
  year         = {2024}
}

@inproceedings{BehnezhadRR24,
  author       = {Soheil Behnezhad and
                  Mohammad Roghani and
                  Aviad Rubinstein},
  editor       = {Bojan Mohar and
                  Igor Shinkar and
                  Ryan O'Donnell},
  title        = {Approximating Maximum Matching Requires Almost Quadratic Time},
  booktitle    = {Proceedings of the 56th Annual {ACM} Symposium on Theory of Computing,
                  {STOC} 2024, Vancouver, BC, Canada, June 24-28, 2024},
  pages        = {444--454},
  publisher    = {{ACM}},
  year         = {2024},
  url          = {https://doi.org/10.1145/3618260.3649785},
  doi          = {10.1145/3618260.3649785},
  bibsource    = {dblp computer science bibliography, https://dblp.org}
}

@inproceedings{BehnezhadRR23b,
  author       = {Soheil Behnezhad and
                  Mohammad Roghani and
                  Aviad Rubinstein},
  title        = {Local Computation Algorithms for Maximum Matching: New Lower Bounds},
  booktitle    = {64th {IEEE} Annual Symposium on Foundations of Computer Science, {FOCS}
                  2023, Santa Cruz, CA, USA, November 6-9, 2023},
  pages        = {2322--2335},
  publisher    = {{IEEE}},
  year         = {2023},
  url          = {https://doi.org/10.1109/FOCS57990.2023.00143},
  doi          = {10.1109/FOCS57990.2023.00143},
  bibsource    = {dblp computer science bibliography, https://dblp.org}
}

@inproceedings{BehnezhadRR23a,
  author       = {Soheil Behnezhad and
                  Mohammad Roghani and
                  Aviad Rubinstein},
  editor       = {Barna Saha and
                  Rocco A. Servedio},
  title        = {Sublinear Time Algorithms and Complexity of Approximate Maximum Matching},
  booktitle    = {Proceedings of the 55th Annual {ACM} Symposium on Theory of Computing,
                  {STOC} 2023, Orlando, FL, USA, June 20-23, 2023},
  pages        = {267--280},
  publisher    = {{ACM}},
  year         = {2023},
  url          = {https://doi.org/10.1145/3564246.3585231},
  doi          = {10.1145/3564246.3585231},
  bibsource    = {dblp computer science bibliography, https://dblp.org}
}

@inproceedings{BhattacharyaKS-STOC23,
  author    = {Sayan Bhattacharya and Peter Kiss and Thatchaphol Saranurak},
  title     = {Sublinear Algorithms for $(1.5 + \epsilon)$-Approximate Matching},
  year      = {2023},
  booktitle = {Proceedings of the 55th {ACM} Symposium on Theory of Computing, {STOC}
               2023, Orlando, Florida}
}

@inproceedings{gupta2015greedy,
  title={Greedy algorithms for Steiner forest},
  author={Gupta, Anupam and Kumar, Amit},
  booktitle={Proceedings of the forty-seventh annual ACM symposium on Theory of Computing},
  pages={871--878},
  year={2015}
}

@article{chazelle2005approximating,
  title={Approximating the minimum spanning tree weight in sublinear time},
  author={Chazelle, Bernard and Rubinfeld, Ronitt and Trevisan, Luca},
  journal={SIAM Journal on computing},
  volume={34},
  number={6},
  pages={1370--1379},
  year={2005},
  publisher={SIAM}
}

@inproceedings{mahabadi2025sublinear,
  title={Sublinear Metric Steiner Tree via Improved Bounds for Set Cover},
  author={Mahabadi, Sepideh and Roghani, Mohammad and Tarnawski, Jakub and Vakilian, Ali},
  booktitle={16th Innovations in Theoretical Computer Science Conference (ITCS 2025)},
  pages={74--1},
  year={2025},
  organization={Schloss Dagstuhl--Leibniz-Zentrum f{\"u}r Informatik}
}

@inproceedings{MahabadiRT25,
  author       = {Sepideh Mahabadi and
                  Mohammad Roghani and
                  Jakub Tarnawski},
  editor       = {Keren Censor{-}Hillel and
                  Fabrizio Grandoni and
                  Joel Ouaknine and
                  Gabriele Puppis},
  title        = {A 0.51-Approximation of Maximum Matching in Sublinear $n^{1.5}$
                  Time},
  booktitle    = {52nd International Colloquium on Automata, Languages, and Programming,
                  {ICALP} 2025, July 8-11, 2025, Aarhus, Denmark},
  series       = {LIPIcs},
  volume       = {334},
  pages        = {116:1--116:17},
  publisher    = {Schloss Dagstuhl - Leibniz-Zentrum f{\"{u}}r Informatik},
  year         = {2025},
  url          = {https://doi.org/10.4230/LIPIcs.ICALP.2025.116},
  doi          = {10.4230/LIPICS.ICALP.2025.116},
  bibsource    = {dblp computer science bibliography, https://dblp.org}
}

@article{abrrfocs25,
  author       = {Amir Azarmehr and
                  Soheil Behnezhad and
                  Mohammad Roghani and
                  Aviad Rubinstein},
  title        = {Tight Pair Query Lower Bounds for Matching and Earth Mover’s Distance},
  journal    = {Proceedings of 66th {IEEE} Annual Symposium on Foundations of Computer Science, {FOCS}
                  2025, Sydney, Australia, December 14-17, 2025, To Appear},
  pages        = {},
  publisher    = {{IEEE}},
  year         = {2025},
  url          = {},
  doi          = {}
}

@inproceedings{indyk1999sublinear,
  title={Sublinear time algorithms for metric space problems},
  author={Indyk, Piotr},
  booktitle={Proceedings of the thirty-first annual ACM symposium on Theory of computing},
  pages={428--434},
  year={1999}
}

@article{agrawal1995trees,
  title={When Trees Collide: An Approximation Algorithm for the Generalized Steiner Problem on Networks},
  author={Agrawal, Ajit and Klein, Philip and Ravi, R},
  journal={SIAM Journal on Computing},
  volume={24},
  number={3},
  pages={440--456},
  year={1995},
  publisher={SIAM}
}

@article{goemans1995general,
  title={A general approximation technique for constrained forest problems},
  author={Goemans, Michel X and Williamson, David P},
  journal={SIAM Journal on Computing},
  volume={24},
  number={2},
  pages={296--317},
  year={1995},
  publisher={SIAM}
}

@inproceedings{gross2018local,
  title={A Local-Search Algorithm for Steiner Forest},
  author={Gro{\ss}, Martin and Gupta, Anupam and Kumar, Amit and Matuschke, Jannik and Schmidt, Daniel R and Schmidt, Melanie and Verschae, Jos{\'e}},
  booktitle={9th Innovations in Theoretical Computer Science Conference (ITCS 2018)},
  year={2018},
  organization={Schloss Dagstuhl--Leibniz-Zentrum f{\"u}r Informatik}
}

@article{czumaj2024streaming,
  title={Streaming algorithms for geometric steiner forest},
  author={Czumaj, Artur and Jiang, Shaofeng H-C and Krauthgamer, Robert and Vesel{\`y}, Pavel},
  journal={ACM Transactions on Algorithms},
  volume={20},
  number={4},
  pages={1--38},
  year={2024},
  publisher={ACM New York, NY}
}

@inproceedings{bamas2022improved,
  title={An Improved Analysis of Greedy for Online Steiner Forest},
  author={Bamas, {\'E}tienne and Drygala, Marina and Maggiori, Andreas},
  booktitle={Proceedings of the 2022 Annual ACM-SIAM Symposium on Discrete Algorithms (SODA)},
  pages={3202--3229},
  year={2022},
  organization={SIAM}
}

@inproceedings{ghalami2022parallel,
  title={A parallel approximation algorithm for the steiner forest problem},
  author={Ghalami, Laleh and Grosu, Daniel},
  booktitle={2022 30th Euromicro International Conference on Parallel, Distributed and Network-based Processing (PDP)},
  pages={47--54},
  year={2022},
  organization={IEEE}
}

@article{whitworth,
  author    = {W. Allen Whitworth},
  title     = {Arrangements of $m$ things of one sort and $n$ things of another sort under certain conditions of priority},
  journal   = {Messenger of Mathematics},
  volume    = {8},
  pages     = {105--114},
  year      = {1878}
}

@article{bertrand1887solution,
  title={Solution d’un probleme},
  author={Bertrand, Joseph},
  journal={CR Acad. Sci. Paris},
  volume={105},
  number={1887},
  pages={369},
  year={1887}
}

@inproceedings{zuckerman2006linear,
  title={Linear degree extractors and the inapproximability of max clique and chromatic number},
  author={Zuckerman, David},
  booktitle={Proceedings of the thirty-eighth annual ACM symposium on Theory of computing},
  pages={681--690},
  year={2006}
}

@inproceedings{hastad1996clique,
  title={Clique is Hard to Approximate Within $n^{1-\epsilon}$},
  author={Hastad, Johan},
  booktitle={Proceedings of 37th Conference on Foundations of Computer Science},
  pages={627--636},
  year={1996},
  organization={IEEE}
}

@inproceedings{luby1985simple,
  title={A simple parallel algorithm for the maximal independent set problem},
  author={Luby, Michael},
  booktitle={Proceedings of the seventeenth annual ACM symposium on Theory of computing},
  pages={1--10},
  year={1985}
}

@article{alon1986fast,
  title={A fast and simple randomized parallel algorithm for the maximal independent set problem},
  author={Alon, Noga and Babai, L{\'a}szl{\'o} and Itai, Alon},
  journal={Journal of algorithms},
  volume={7},
  number={4},
  pages={567--583},
  year={1986},
  publisher={Elsevier}
}

@inproceedings{ghaffari2016improved,
  title={An improved distributed algorithm for maximal independent set},
  author={Ghaffari, Mohsen},
  booktitle={Proceedings of the twenty-seventh annual ACM-SIAM symposium on Discrete algorithms},
  pages={270--277},
  year={2016},
  organization={SIAM}
}

@inproceedings{ghaffari2018improved,
  title={Improved massively parallel computation algorithms for mis, matching, and vertex cover},
  author={Ghaffari, Mohsen and Gouleakis, Themis and Konrad, Christian and Mitrovi{\'c}, Slobodan and Rubinfeld, Ronitt},
  booktitle={Proceedings of the 2018 ACM Symposium on Principles of Distributed Computing},
  pages={129--138},
  year={2018}
}

@inproceedings{ghaffari2022local,
  title={Local computation of maximal independent set},
  author={Ghaffari, Mohsen},
  booktitle={2022 IEEE 63rd Annual Symposium on Foundations of Computer Science (FOCS)},
  pages={438--449},
  year={2022},
  organization={IEEE}
}

@inproceedings{RubinfeldTVX11,
  author       = {Ronitt Rubinfeld and
                  Gil Tamir and
                  Shai Vardi and
                  Ning Xie},
  title        = {Fast Local Computation Algorithms},
  booktitle    = {Innovations in Computer Science - {ICS}},
  pages        = {223--238},
  publisher    = {Tsinghua University Press},
  year         = {2011}
}

@inproceedings{alon2012space,
  title={Space-efficient local computation algorithms},
  author={Alon, Noga and Rubinfeld, Ronitt and Vardi, Shai and Xie, Ning},
  booktitle={Proceedings of the twenty-third annual ACM-SIAM symposium on Discrete Algorithms},
  pages={1132--1139},
  year={2012},
  organization={SIAM}
}

@article{parter2019local,
  title={Local computation algorithms for spanners},
  author={Parter, Merav and Rubinfeld, Ronitt and Vakilian, Ali and Yodpinyanee, Anak},
  journal={arXiv preprint arXiv:1902.08266},
  year={2019}
}

@misc{chekuri2025streaming,
      title={Streaming Algorithms for Network Design}, 
      author={Chandra Chekuri and Rhea Jain and Sepideh Mahabadi and Ali Vakilian},
      year={2025},
      eprint={2503.00712},
      archivePrefix={arXiv},
      primaryClass={cs.DS},
      url={https://arxiv.org/abs/2503.00712}, 
}

@article{arviv2023improved,
  title={Improved Local Computation Algorithms for Constructing Spanners},
  author={Arviv, Rubi and Chung, Lily and Levi, Reut and Pyne, Edward},
  journal={Approximation, Randomization, and Combinatorial Optimization. Algorithms and Techniques},
  year={2023}
}

@inproceedings{lenzen2018centralized,
  title={A centralized local algorithm for the sparse spanning graph problem},
  author={Lenzen, Christoph and Levi, Reut},
  booktitle={45th International Colloquium on Automata, Languages, and Programming},
  pages={1--47},
  year={2018},
  organization={Schloss Dagstuhl}
}

@inproceedings{jin2024streaming,
  title={Streaming Algorithms for Connectivity Augmentation},
  author={Jin, Ce and Kapralov, Michael and Mahabadi, Sepideh and Vakilian, Ali},
  booktitle={51st International Colloquium on Automata, Languages, and Programming (ICALP 2024)},
  pages={93--1},
  year={2024},
  organization={Schloss Dagstuhl--Leibniz-Zentrum f{\"u}r Informatik}
}

@inproceedings{AilonCN05,
  author       = {Nir Ailon and
                  Moses Charikar and
                  Alantha Newman},
  editor       = {Harold N. Gabow and
                  Ronald Fagin},
  title        = {Aggregating inconsistent information: ranking and clustering},
  booktitle    = {Proceedings of the 37th Annual {ACM} Symposium on Theory of Computing,
                  Baltimore, MD, USA, May 22-24, 2005},
  pages        = {684--693},
  publisher    = {{ACM}},
  year         = {2005},
  url          = {https://doi.org/10.1145/1060590.1060692},
  doi          = {10.1145/1060590.1060692},
  bibsource    = {dblp computer science bibliography, https://dblp.org}
}

@inproceedings{BehnezhadDHSS19,
  author       = {Soheil Behnezhad and
                  Mahsa Derakhshan and
                  MohammadTaghi Hajiaghayi and
                  Cliff Stein and
                  Madhu Sudan},
  editor       = {David Zuckerman},
  title        = {Fully Dynamic Maximal Independent Set with Polylogarithmic Update
                  Time},
  booktitle    = {60th {IEEE} Annual Symposium on Foundations of Computer Science, {FOCS}
                  2019, Baltimore, Maryland, USA, November 9-12, 2019},
  pages        = {382--405},
  publisher    = {{IEEE} Computer Society},
  year         = {2019},
  url          = {https://doi.org/10.1109/FOCS.2019.00032},
  doi          = {10.1109/FOCS.2019.00032},
  bibsource    = {dblp computer science bibliography, https://dblp.org}
}

@inproceedings{BehnezhadCMT22,
  author       = {Soheil Behnezhad and
                  Moses Charikar and
                  Weiyun Ma and
                  Li{-}Yang Tan},
  title        = {Almost 3-Approximate Correlation Clustering in Constant Rounds},
  booktitle    = {63rd {IEEE} Annual Symposium on Foundations of Computer Science, {FOCS}
                  2022, Denver, CO, USA, October 31 - November 3, 2022},
  pages        = {720--731},
  publisher    = {{IEEE}},
  year         = {2022},
  url          = {https://doi.org/10.1109/FOCS54457.2022.00074},
  doi          = {10.1109/FOCS54457.2022.00074},
  bibsource    = {dblp computer science bibliography, https://dblp.org}
}

@inproceedings{DalirrooyfardMS24,
author = {Dalirrooyfard, Mina and Makarychev, Konstantin and Mitrovi\'{c}, Slobodan},
title = {Pruned Pivot: correlation clustering algorithm for dynamic, parallel, and local computation models},
year = {2024},
publisher = {JMLR.org},
booktitle = {Proceedings of the 41st International Conference on Machine Learning},
articleno = {394},
numpages = {22},
location = {Vienna, Austria},
series = {ICML'24}
}

@article{BrandtFU18,
  author       = {Sebastian Brandt and
                  Manuela Fischer and
                  Jara Uitto},
  title        = {Matching and {MIS} for Uniformly Sparse Graphs in the Low-Memory {MPC}
                  Model},
  journal      = {CoRR},
  volume       = {abs/1807.05374},
  year         = {2018},
  url          = {http://arxiv.org/abs/1807.05374},
  eprinttype    = {arXiv},
  eprint       = {1807.05374},
  bibsource    = {dblp computer science bibliography, https://dblp.org}
}

@inproceedings{GhaffariGJ20,
  author       = {Mohsen Ghaffari and
                  Christoph Grunau and
                  Ce Jin},
  editor       = {Hagit Attiya},
  title        = {Improved {MPC} Algorithms for MIS, Matching, and Coloring on Trees
                  and Beyond},
  booktitle    = {34th International Symposium on Distributed Computing, {DISC} 2020,
                  October 12-16, 2020, Virtual Conference},
  series       = {LIPIcs},
  volume       = {179},
  pages        = {34:1--34:18},
  publisher    = {Schloss Dagstuhl - Leibniz-Zentrum f{\"{u}}r Informatik},
  year         = {2020},
  url          = {https://doi.org/10.4230/LIPIcs.DISC.2020.34},
  doi          = {10.4230/LIPICS.DISC.2020.34},
  bibsource    = {dblp computer science bibliography, https://dblp.org}
}

@inproceedings{LattanziMSV11,
author = {Lattanzi, Silvio and Moseley, Benjamin and Suri, Siddharth and Vassilvitskii, Sergei},
title = {Filtering: a method for solving graph problems in MapReduce},
year = {2011},
isbn = {9781450307437},
publisher = {Association for Computing Machinery},
address = {New York, NY, USA},
booktitle = {Proceedings of the Twenty-Third Annual ACM Symposium on Parallelism in Algorithms and Architectures},
pages = {85–94},
numpages = {10},
location = {San Jose, California, USA},
series = {SPAA '11}
}

@inbook{GhaffariU19,
author = {Mohsen Ghaffari and Jara Uitto},
title = {Sparsifying Distributed Algorithms with Ramifications in Massively Parallel Computation and Centralized Local Computation},
booktitle = {Proceedings of the 2019 Annual ACM-SIAM Symposium on Discrete Algorithms (SODA)},
chapter = {},
year = {2019},
pages = {1636-1653},
doi = {10.1137/1.9781611975482.99},
URL = {https://epubs.siam.org/doi/abs/10.1137/1.9781611975482.99}
}

@article{Onak18,
  author       = {Krzysztof Onak},
  title        = {Round Compression for Parallel Graph Algorithms in Strongly Sublinear
                  Space},
  journal      = {CoRR},
  volume       = {abs/1807.08745},
  year         = {2018},
  url          = {http://arxiv.org/abs/1807.08745},
  eprinttype    = {arXiv},
  eprint       = {1807.08745},
  bibsource    = {dblp computer science bibliography, https://dblp.org}
}

@inproceedings{AssadiOSS19,
  author       = {Sepehr Assadi and
                  Krzysztof Onak and
                  Baruch Schieber and
                  Shay Solomon},
  editor       = {Timothy M. Chan},
  title        = {Fully Dynamic Maximal Independent Set with Sublinear in n Update Time},
  booktitle    = {Proceedings of the Thirtieth Annual {ACM-SIAM} Symposium on Discrete
                  Algorithms, {SODA} 2019, San Diego, California, USA, January 6-9,
                  2019},
  pages        = {1919--1936},
  publisher    = {{SIAM}},
  year         = {2019},
  url          = {https://doi.org/10.1137/1.9781611975482.116},
  doi          = {10.1137/1.9781611975482.116},
  bibsource    = {dblp computer science bibliography, https://dblp.org}
}

@inproceedings{ChechikZ19,
  author       = {Shiri Chechik and
                  Tianyi Zhang},
  editor       = {David Zuckerman},
  title        = {Fully Dynamic Maximal Independent Set in Expected Poly-Log Update
                  Time},
  booktitle    = {60th {IEEE} Annual Symposium on Foundations of Computer Science, {FOCS}
                  2019, Baltimore, Maryland, USA, November 9-12, 2019},
  pages        = {370--381},
  publisher    = {{IEEE} Computer Society},
  year         = {2019},
  url          = {https://doi.org/10.1109/FOCS.2019.00031},
  doi          = {10.1109/FOCS.2019.00031},
  bibsource    = {dblp computer science bibliography, https://dblp.org}
}

@article{LeviRY17,
  author       = {Reut Levi and
                  Ronitt Rubinfeld and
                  Anak Yodpinyanee},
  title        = {Local Computation Algorithms for Graphs of Non-constant Degrees},
  journal      = {Algorithmica},
  volume       = {77},
  number       = {4},
  pages        = {971--994},
  year         = {2017},
  url          = {https://doi.org/10.1007/s00453-016-0126-y},
  doi          = {10.1007/S00453-016-0126-Y},
  bibsource    = {dblp computer science bibliography, https://dblp.org}
}

@article{ReingoldV16,
  author       = {Omer Reingold and
                  Shai Vardi},
  title        = {New techniques and tighter bounds for local computation algorithms},
  journal      = {J. Comput. Syst. Sci.},
  volume       = {82},
  number       = {7},
  pages        = {1180--1200},
  year         = {2016},
  url          = {https://doi.org/10.1016/j.jcss.2016.05.007},
  doi          = {10.1016/J.JCSS.2016.05.007}
}

@inbook{GhaffariH21,
author = {Mohsen Ghaffari and Bernhard Haeupler},
title = {A Time-Optimal Randomized Parallel Algorithm for MIS},
booktitle = {Proceedings of the 2021 ACM-SIAM Symposium on Discrete Algorithms (SODA)},
chapter = {},
year={2021},
pages = {2892-2903},
doi = {10.1137/1.9781611976465.172},
URL = {https://epubs.siam.org/doi/abs/10.1137/1.9781611976465.172}
}

@inproceedings{czumaj2003sublinear,
  title={Sublinear-time approximation of Euclidean minimum spanning tree.},
  author={Czumaj, Artur and Erg{\"u}n, Funda and Fortnow, Lance and Magen, Avner and Newman, Ilan and Rubinfeld, Ronitt and Sohler, Christian and others},
  booktitle={SODA},
  pages={813--822},
  year={2003}
}

@article{FischerA19,
author = {Fischer, Manuela and Noever, Andreas},
title = {Tight Analysis of Parallel Randomized Greedy MIS},
year = {2019},
issue_date = {January 2020},
publisher = {Association for Computing Machinery},
address = {New York, NY, USA},
volume = {16},
number = {1},
issn = {1549-6325},
url = {https://doi.org/10.1145/3326165},
doi = {10.1145/3326165},
journal = {ACM Trans. Algorithms},
month = dec,
articleno = {6},
numpages = {13}
}

@inproceedings{BlellochFS12,
author = {Blelloch, Guy E. and Fineman, Jeremy T. and Shun, Julian},
title = {Greedy sequential maximal independent set and matching are parallel on average},
year = {2012},
isbn = {9781450312134},
publisher = {Association for Computing Machinery},
address = {New York, NY, USA},
url = {https://doi.org/10.1145/2312005.2312058},
doi = {10.1145/2312005.2312058},
booktitle = {Proceedings of the Twenty-Fourth Annual ACM Symposium on Parallelism in Algorithms and Architectures},
pages = {308–317},
numpages = {10},
location = {Pittsburgh, Pennsylvania, USA},
series = {SPAA '12}
}

@inproceedings{ahmadi2025breaking,
  title={Breaking a Long-Standing Barrier: 2-$\varepsilon$ Approximation for Steiner Forest},
  author={Ahmadi, Ali and Gholami, Iman and Hajiaghayi, MohammadTaghi and Jabbarzade, Peyman and Mahdavi, Mohammad},
    booktitle={2025 IEEE 66th Annual Symposium on Foundations of Computer Science (FOCS)}, 
  year={2025}
}

\end{document}